\documentclass{lmcs}
\usepackage{lineno}
\usepackage{amssymb}
\usepackage{tikz}
\usepackage{amsthm}
\theoremstyle{plain}\newtheorem*{proviso}{Proviso}

\theoremstyle{plain}\newtheorem{claim}[thm]{Claim}
\newenvironment{claimproof}[1]{\par\noindent \emph{Proof:}\space#1}{\hfill \renewcommand{\qedsymbol}{\ensuremath{\blacksquare}} \qedsymbol}
 %TODO mandatory; please add comma-separated list of keywords

\newcommand{\poly}{\operatorname{poly}}

\newcommand{\Enum}{\operatorname{Enum}}
\newcommand{\dv}{\operatorname{dv}}

\begin{document}
\title{Towards Approximate Query Enumeration with Sublinear Preprocessing Time}

\author{Isolde Adler}
\address{University of Leeds, School of Computing, Leeds, UK}
\email{i.m.adler@leeds.ac.uk}%TODO mandatory, please use full name; only 1 author per \author macro; first two parameters are mandatory, other parameters can be empty. Please provide at least the name of the affiliation and the country. The full address is optional

\author{Polly Fahey}
\address{University of Leeds, School of Computing, Leeds, UK}
\email{mm11pf@leeds.ac.uk}
\begin{abstract}
This paper aims at providing extremely efficient algorithms for approximate query enumeration on sparse databases, that come with performance and accuracy guarantees. We introduce a new model for approximate query enumeration on classes of relational databases of bounded degree. We first prove that on databases of bounded degree any \emph{local} first-order definable query that has a sufficiently large answer set can be enumerated approximately with constant delay after a preprocessing phase with \emph{constant} running time. 
%Where we say a first-order query is `local' if given any bounded degree database and tuple it can be decided by only looking at the local (fixed radius) neighbourhood around the tuple whether the tuple is an answer to the query for the database or not.
We extend this, showing that on databases of bounded tree-width and bounded degree, every query that is expressible in first-order logic and has a sufficiently large answer set can be enumerated approximately with constant delay after a preprocessing phase with \emph{sublinear} (more precisely, \emph{polylogarithmic}) running time. 
	
Durand and Grandjean (ACM Transactions on Computational Logic 2007) proved that \emph{exact} enumeration of first-order queries on databases of bounded degree can be done with constant delay after a preprocessing phase with running time \emph{linear} in the size of the input database.	Hence we achieve a significant speed-up in the preprocessing phase. Since sublinear running time does not allow reading the whole input database even once, sacrificing some accuracy is inevitable for our speed-up. Nevertheless, our enumeration algorithm comes with guarantees:
With high probability, (1) only tuples are enumerated that are answers to the query or `close' to being answers to the query, and (2) if the proportion of tuples that are answers to the query is sufficiently large, then all answers will be enumerated. Here the notion of `closeness' is a tuple edit distance in the input database.  For local first-order queries, only actual answers are enumerated, strengthening (1). Moreover, both the `closeness' and the proportion required in (2) are controllable. Our algorithms only access the input database by sampling local parts, in a distributed fashion.  

While our preprocessing phase is simpler than the preprocessing phase for the exact algorithm, our enumeration phase is more involved, as we push parts of the computation into the enumeration phase, allowing us to keep on enumerating answers.

We combine methods from property testing of bounded degree graphs with logic and query enumeration, which we believe can inspire further research.  
\end{abstract}

\maketitle

\section{Introduction}
Given the ubiquity and sheer size of stored data nowadays, there is an immense need for
	highly efficient algorithms to extract information from the data.
When the input data is huge, many algorithms 
that are traditionally classified as `efficient' become impractical. 
	Hence in practice often heuristics are used, at the price of losing control over
	the quality of the computed information. In many application areas however, such as 
	aviation, security, medicine, and research,
	accuracy guarantees regarding the computed output are crucial. 

We address this by taking a step towards 
foundations for Approximate Query Processing~\cite{chaudhuri2017approximate}.
We provide a new model for 
\emph{approximately} enumerating the set of answers to queries on relational databases. 
%We also extend this to approximate \emph{query membership testing} and 
%\emph{counting}.
%We also extend our results
%to approximate \emph{query membership testing} and \emph{counting}.
This enables us to decrease the running time significantly 
compared to traditional algorithms 
while providing probabilistic accuracy guarantees.
 
\paragraph*{\textbf{Query enumeration.}} Query evaluation plays a central role in databases systems, and in the past decades it has received huge attention both from practical and theoretical perspectives. 
One of the central problems is \emph{query enumeration}. Here we are given a 
database $\mathcal D$ and a query $q$, and the goal is to compute the set 
$q(\mathcal D)$ of all
answers to $q$ on $\mathcal D$. However, the set $q(\mathcal D)$
could be exponential in the number of free variables of $q$, and even bigger than $\mathcal D$, 
hence the total running time required to enumerate all answers may not be a meaningful 
complexity measure. 
Taking this into account, models for query enumeration distinguish two phases,
a \emph{preprocessing phase}, and an \emph{enumeration phase}.
Typically, in the preprocessing phase some form of data structure  
is computed from $\mathcal D$ and $q$, in such a way that in the enumeration phase
all answers in $q(\mathcal D)$ can be enumerated (without repetition) with only a small delay between any two consecutive answers. We focus on data complexity, i.\,e.\ we  
regard the query as being fixed, and the database being the input.
\emph{Efficiency} is measured both in terms of
the running time of the preprocessing phase and the \emph{delay}, i.\,e.\ the 
maximum time between the output of any two consecutive answers. 
For the delay we can hope for constant time at best, independent of the
size of the database. For the preprocessing
phase, the best we can hope for regarding exact algorithms is linear time. 

Recent research has been very successful in providing exact enumeration algorithms for 
first-order queries on \emph{sparse} relational databases. In 2007, Durand and Grandjean
showed that on relational databases of bounded degree, every first-order query can be enumerated
with constant delay after a linear time preprocessing phase~\cite{durand2007first}. 
This result triggered a number of 
papers~\cite{KazanaS11,DurandSS14, SegoufinV17},  
culminating in Schweikardt, Segoufin and Vigny's result that
on \emph{nowhere dense} databases, first-order queries can be enumerated with constant delay
after a pseudo-linear time preprocessing phase~\cite{SchweikardtSV18}.

\subparagraph*{\textbf{Our contributions.}}
In this paper we aim at \emph{sublinear time} preprocessing and constant delay 
in the enumeration phase. We consider databases $\mathcal D$ of bounded degree $d$, i.\,e.\
every element of the domain appears in at most $d$ tuples in relations of $\mathcal D$, and we identify conditions under which first-order definable queries can be enumerated approximately
with constant delay after a sublinear preprocessing phase. We consider two different categories of first-order definable queries, \emph{local} and \emph{general} (including \emph{non-local}) queries. A first-order query is \emph{local} if, given any bounded degree database and tuple, it can be decided by only looking at the local (fixed radius) neighbourhood around the tuple whether the tuple is an answer to the query for the database. We show the following.

 \emph{On input databases of bounded degree, every (fixed) local first-order definable query can be enumerated approximately with constant delay after a constant time preprocessing phase (Theorem \ref{thm: local queries enumeration}).}

 \emph{On input databases of bounded degree and bounded tree-width, every (fixed) first-order definable query can be enumerated approximately with constant delay after a sublinear time preprocessing phase (Theorem \ref{thm:non-local enumeration main}).}

We also give generalisations of the two theorems above (Theorems \ref{theorem: strengthened local main}, \ref{theorem: strengthened non local main} and \ref{theorem: strengthened main hanf}) and applications of our approach to further computational problems on databases (Theorems \ref{query membership testing} and \ref{counting theorem}), which we will discuss below.

First, let us give some more details.
For any local first-order query $q$, bounded degree database $\mathcal{D}$ and tuple $\bar{a}$ from $\mathcal{D}$ it can be decided in constant time whether $\bar{a}$ is an answer to $q$ on $\mathcal{D}$ (Lemma \ref{lemma: local membership testing}). Using this fact, we show that for any fixed local first-order definable query $q(\bar{x})$ with $|\bar{x}| =:k$ and $\gamma \in (0,1)$, there exists an enumeration algorithm with constant preprocessing time and constant delay, that is given a bounded degree database $\mathcal{D}$ with domain of size $n$ as input and does the following. It enumerates a set of tuples that are answers to $q$ on $\mathcal{D}$, and with high probability it enumerates \emph{all} answers to $q$ on $\mathcal{D}$ if the size of the answer set of $q$ on $\mathcal{D}$ is larger than $\gamma n^k$ (i.e. the number of answers to the query is larger than a fixed fraction of the total possible number of answers). 

Towards reducing the minimum size of the answer set required to enumerate \emph{all} answers to the query, we show we actually only require size $ \gamma n^c$, where $c$ is the maximum number of connected components in the neighbourhood (of some fixed radius) of an answer to $q$ (Theorem \ref{theorem: strengthened local main}). We argue that in practice, $c$ can be expected to be low for natural queries.

If a first-order query $q$ is non-local, then for a database $\mathcal{D}$ and a tuple $\bar{a}$, we can no longer decided in constant time whether $\bar{a}$ is an answer to $q$ on $\mathcal{D}$. However, using Hanf-locality of first-order logic \cite{Hanf1965} and a result from the area of property testing, we can approximately enumerate any first-order definable query on bounded degree and bounded tree-width databases with polylogarithmic preprocessing time and constant delay (Theorem \ref{thm:non-local enumeration main}). Let us now explain our notion of approximation, which is based on neighbourhood types.

For $d\in \mathbb N$, let $\mathbf C$ be a class of databases 
of degree at most $d$  over a fixed finite schema. %$\sigma$.
Let $\mathcal D\in \mathbf C$, let $r\in \mathbb N$ and let $a$ be an element of the 
domain of $\mathcal D$. The \emph{$r$-neighbourhood type} of $a$ in $\mathcal D$ is
the isomorphism type of the sub-database of $\mathcal D$ induced
by all elements of the domain whose distance to $a$ (in the underlying graph of $\mathcal{D}$)
is at most $r$, expanded by $a$. The element $a$ is called the \emph{centre}. This can be extended to define the $r$-neighbourhood type of a \emph{tuple} $\bar a$ in $\mathcal D$,
by considering the isomorphism type of the sub-database induced by the union of the
$r$-neighbourhoods of all components of $\bar a$, expanded by $\bar a$. We call
such an isomorphism type an \emph{$r$-type (with $|\bar a|$ centres)}. 
Given a database query $q(\bar x)$  
with $|\bar x|=:k$ and a database $\mathcal D$ with domain of size $n$ we say that a tuple $\bar{a}$ from $\mathcal D$ is \emph{$\epsilon$-close to being an answer of $q$ on $\mathcal D$ and $\mathbf C$}, if $\mathcal{D}$ can be modified with tuple modifications (insertions and deletions) into a database $\mathcal{D'} \in \mathbf{C}$ with at most $\epsilon d n$ modifications, such that $\bar{a}$ is an answer of $q$ on $\mathcal{D'} $ and $\bar{a}$ has the same $r$-neighbourhood type (for some $r$) in $\mathcal{D'} $ and $\mathcal{D} $. We let $q(\mathcal{D},\mathbf{C}, \epsilon)$ be the set of $k$-tuples $\bar a$ 
of elements of $\mathcal D$ that are $\epsilon$-close to being an answer of $q$ on $\mathcal D$ and $\mathbf C$. Note that for any local first-order query $q$, $q(\mathcal{D},\mathbf{C}, \epsilon)=q(\mathcal{D})$.

We say that the enumeration problem $\Enum_{\mathbf{C}}(q)$
for $q$ on $\mathbf C$ can be solved \emph{approximately} with preprocessing time $H(n)$ and
constant delay for answer threshold function $f(n)$, if for every $\epsilon \in (0,1]$, there exists an algorithm, 
which is given oracle access to an input database 
$\mathcal{D} \in \mathbf{C}$ (for each given element of the domain, the
tester can query the oracle for tuples in any of the relations containing the element, and
we assume that oracle queries are answered in constant time), and is given the number $n$ of elements of the domain, 
and proceeds in two phases. First, a preprocessing phase that runs in time $H(n)$, followed by
an enumeration phase where a set $S$ of pairwise distinct tuples is enumerated, 
with constant delay between any two consecutive tuples.
In addition, we require that with probability at least $2/3$,
%\begin{itemize}
		(1) $S \subseteq q(\mathcal{D}) \cup q(\mathcal{D},\mathbf{C}, \epsilon)$, and
	   (2) if $|q(\mathcal{D})| \geq f(n)$, then $q(\mathcal{D}) \subseteq S$.
%\end{itemize}

We consider database queries that are expressible in first-order logic. Note that our notion of approximation is designed specifically for first-order queries and sparse databases and for other classes of queries and input databases alternative models may be necessary. We prove that for every first-order query $q(\bar{x})$ with $|\bar{x}|=k$ the problem $\Enum_{\mathbf C_d^t}(q)$ (where $\mathbf C_d^t$ is the class of
all databases of $d$-bounded degree and $t$-bounded tree-width) can be solved approximately with polylogarithmic preprocessing time and constant delay with answer threshold function $f(n)=\gamma n^k$ for any $\gamma \in (0,1)$ (Theorem~\ref{thm:non-local enumeration main}). As with local queries, we further prove that we can actually reduce the answer threshold function to $f(n)=\gamma n^c$ where $c \leq k$ is the maximum number of connected components in the neighbourhood (of some fixed radius) of an answer to $q$ (Theorem \ref{theorem: strengthened non local main}).
We also identify a condition that is based on Hanf-locality of first-order logic~\cite{Hanf1965}, which
we call \emph{Hanf-sentence testability}, and we prove a general theorem (Theorem~\ref{theorem: strengthened main hanf}),
that for every first-order query $q(\bar{x})$ with $|\bar{x}|=k$ that is Hanf-sentence testable on $\mathbf C$
in time $H(n)$, the problem $\Enum_{\mathbf{C}}(q)$ can be solved approximately 
with preprocessing time $\mathcal{O}(H(n))$ and constant delay for answer threshold function $f(n)=\gamma n^c$ as above.

We illustrate our model throughout the paper with a running example which can be motivated by the problems of subgraph matching and inexact subgraph matching in social and biological networks (e.g. \cite{zhang2009gaddi, tong2007fast}). We show that our running example is Hanf-sentence testable on the class of all bounded degree graphs in constant time, and hence by Theorem \ref{theorem: strengthened main hanf} it can be approximately enumerated with constant preprocessing time.

  Our notion of approximation is based on `structural' closeness and therefore our algorithms are aimed at applications where structural similarity is essential. 
For example, when given a new huge dataset (such as biological datasets or social networks),
a first exploration of the approximate structure with some accuracy guarantees may be
desirable to obtain initial insights quickly. These insights could then e.\,g.\ be used to make 
decisions regarding more time consuming analysis in a follow-up stage.

\subparagraph*{\textbf{Property testing.}}
%Our approach combines of methods from property testing of bounded degree graphs 
%with logic and query enumeration. 
Before sketching the proof idea of Theorem~\ref{thm:non-local enumeration main}, let us give some background on property testing. Property testing aims at providing highly efficient algorithms that 
derive \emph{global} information on the
structure of the input by only exploring a small number of \emph{local} parts of it.
These algorithms are randomised and allow for a small error. Nevertheless, they 
come with guarantees regarding both
the quality of the solution and efficiency. Typically, the algorithms
 only look at a constant number of small parts of the input,
and they run in constant or sublinear time. Even for problems that allow linear time
exact algorithms, such as graph connectivity, reducing the running time
(while sacrificing some accuracy) may become crucial if the networks are huge.
Property testing can be seen as solving \emph{relaxed decision problems}. 
Instead of deciding whether a given input has a certain property, 
the goal is to determine with high probability correctly, whether the input has the
property or is far from having it.
Formally, a \emph{property} $\mathbf P$ is an isomorphism closed class of relational databases. 
For example, each Boolean database query $q$ defines a property
$\mathbf P_q$, the class of all databases satisfying $q$.
A \emph{property testing algorithm} (\emph{tester}, for short) for $\mathbf P$ determines whether a given database $\mathcal D$ has property  $\mathbf P$ (i.\,e.\ whether $\mathcal D$ is a member of $\mathbf P$) or 
is $\epsilon$-far from having $\mathbf P$. 
Testers are randomised and allow for a small constant error probability.
The algorithms are parameterised by a distance measure $\epsilon$, where
the distance measure depends on the model.

Property testing was first 
introduced in~\cite{RubinfeldS96}, in the context of Programme Checking. 
In this paper we build on 
the model for property testing on relational databases of \emph{bounded degree} 
of~\cite{adler2018property}, which is a generalisation of the bounded degree graph 
model~\cite{goldreich2002property}. 
This model assumes a uniform upper bound $d$ on the degree of the input databases. 
For $\epsilon\in [0,1]$, 
a database $\mathcal D $ with domain of size $n$
is \emph{$\epsilon$-close} to satisfying $\mathbf P$, if we can make $\mathcal D$ isomorphic to
a member of  $\mathbf P$ by editing (inserting or removing) at most $\epsilon dn$ 
tuples in relations of
$\mathcal D$ (i.\,e.\ at most an `$\epsilon$-fraction' of the 
maximum possible number $dn$ of tuples in relations). Otherwise, $\mathcal D$ is called \emph{$\epsilon$-far} from $\mathbf P$.
An $\epsilon$-tester receives the size $n$ of the domain of the input, and
and has oracle access to the database.

\subparagraph*{\textbf{Techniques.}}
To give a flavour of our techniques, we sketch the proof idea of Theorem~\ref{thm:non-local enumeration main}. Let $\phi(\bar{x})$ be a
first-order formula with $|\bar{x}|=k$ and let $\mathcal{D}$ be an input database
from the class of databases $\mathbf{C}_d^t$ with bounded degree and bounded tree-width over a fixed finite
schema. In the preprocessing phase, we first compute a formula $\chi$ that is 
equivalent to $\phi$ on $\mathbf{C}_d^t$, and $\chi$ is in a special type of 
Hanf normal form
that groups the Hanf-sentences and sphere formulas together (Lemma~\ref{normal-form lemma}).  
We then run property testers on the sentence parts of $\chi$ and compute a set
$T$ of $r$-types (where $r$ is the Hanf locality radius of $\phi$), that with high probability
for any $\mathcal{D} \in \mathbf{C}_d^t $ and $\bar{a} \in D^k$,
if $\bar{a} \in \phi(\mathcal{D})$, then the $r$-type of $\bar{a}$ in $\mathcal{D}$ is in $T$, and
if $\bar{a} \in D^k \setminus \phi(\mathcal{D},\mathbf{C}_d^t, \epsilon )$, then the $r$-type of $\bar{a}$ in $\mathcal{D}$ is not in $T$.
 In the remainder of the preprocessing
phase we randomly sample a constant number of $k$-tuples of elements of
$\mathcal{D}$ and check whether their $r$-type is in $T$. Assuming
that $|\phi(\mathcal{D})|$ is sufficiently large, with high probability we will have sampled at
least one tuple whose $r$-type is in $T$, and we start the enumeration phase by 
enumerating this tuple. To keep the enumeration going, after each tuple that is enumerated,
we sample 
a constant number of tuples from the input database. To avoid outputting duplicates we keep a record of which tuples we have already seen by using an array that can be updated and read in constant time. Finally, 
with high probability we will see every possible
$k$-tuple of elements of $\mathcal{D}$.

\subparagraph*{\textbf{Further related research.}}
So far, only a small number of results in database theory make use of models from property testing. Chen and Yoshida~\cite{chen2019testability} study the testability of homomorphism inadmissibility in a model which is close to the general graph model (cf.\ e.\,g. \cite{alon2008testing}).
%They access the database with a completion oracle. 
%With the completion oracle given a tuple of elements from the database and a special wildcard symbol they can retrieve a random completion to this tuple or the number of possible completions. 
Ben-Moshe et al. \cite{ben2011detecting} study the testability of near-sortedness (a property of relations that states that most tuples are close to their place in some desired order). 
%In their model they are allowed to ask for the ith tuple. 
Our model differs from both of these, as it relies on a degree bound and uses a different type of oracle access.
A \emph{conjunctive query} (CQ) is a first-order formula constructed from atomic formulas using conjunctions and existential quantification only. CQ evaluation is closely related to solving 
constraint satisfaction problems (CSPs)~\cite{kolaitis2000conjunctive}. CSPs have been studied 
under different models from property testing (\cite{chen2019constant, yoshida2011optimal, alon2003random}). Our work, however, is relevant for more complex queries, as enumerating CQs in our model basically amounts to sampling.

Our work is a step towards approximate enumeration on sparse databases. 
It would be interesting to study approximate enumeration on databases of bounded \emph{average degree}. However, this would require different techniques.

\subparagraph*{\textbf{Organisation.}}
In Section~\ref{section prelim} we introduce notions used throughout the paper. In Section \ref{sec: local and non local queries} we give some useful normal forms of first-order queries along with some results on local first-order queries. In Sections \ref{sec: local queries} and \ref{sec: non local queries} we prove our main theorems on the enumeration of local and general first-order queries respectively. In Section \ref{sec: proof of theorem strengthened local main}, in an attempt to push the boundaries further, we prove strengthened versions of the theorems proved in Sections \ref{sec: local queries} and \ref{sec: non local queries}, showing how the required answer threshold can be reduced. Finally, in Section \ref{sec: further results} we prove a generalisation of our main theorem on approximate enumeration of general first-order queries showing that the assumption of bounded tree-width can be replaced with the weaker assumption of Hanf-sentence testability. We also provide results on approximate membership testing and approximate counting.
%In Section~\ref{section prelim} we introduce notions used throughout the paper. In Section~\ref{sec:enum} we introduce our model for approximate query enumeration on classes of relational databases of bounded degree. 
%Then in Section~\ref{sec:main} we prove our main theorems. 
%We conclude in Section~\ref{sec:conclusion}. Due to space constraints the proofs of statements labelled $(\ast)$ are deferred to the appendix.

\section{Preliminaries}\label{section prelim}

We let $\mathbb{N}$ be the set of natural numbers including $0$, and $\mathbb{N}_{\geq 1} = \mathbb{N} \setminus \{0\}$. For each $n \in \mathbb{N}_{\geq 1}$, we let $[n] = \{1,2,\dots,n\}$.

\subsection*{Databases.}
A \emph{schema} is a finite set $\sigma = \{R_1,\dots,R_{|\sigma|}\}$ of relation names, 
where each $R\in \sigma$ has an \emph{arity} ar$(R) \in \mathbb{N}_{\geq 1}$. 
The \emph{size} of a schema, denoted by $\|\sigma\|$, is the sum of the arities of its relation names. A \emph{database} $\mathcal{D}$ of schema $\sigma$ ($\sigma$-db for short) is of the form
$\mathcal{D} = (D, R_1^{\mathcal{D}}, \dots, R_{|\sigma|}^{\mathcal{D}})$, where $D$ is a finite set, the set
 of \emph{elements} of $\mathcal{D}$, and $R_i^{\mathcal{D}}$ is an ar$(R_i)$-ary relation on $D$.
 The set $D$ is also called the \emph{domain} of $\mathcal{D}$. An \emph{(undirected) graph} $\mathcal{G}$ is a tuple $\mathcal{G} =(V(\mathcal{G}),E(\mathcal{G}))$ where $V(\mathcal{G})$ is a set of \emph{vertices} and $E(\mathcal{G})$ is a set of $2$-element subsets of $V(\mathcal{G})$ (the \emph{edges} of $\mathcal G$). 
For an edge $\{u,v\}\in E(\mathcal G)$ we simply write $uv$.
For a graph $\mathcal G$ with $uv\in E(\mathcal G)$ we let $\mathcal G\setminus uv$ denote the graph obtained from $\mathcal{G}$ by removing the edge $uv$ from $E({\mathcal{G})}.$
An undirected graph can be seen as a $\{E\}$-db, where $E$ is a binary relation name, interpreted by a symmetric, irreflexive relation.

We assume that all databases are linearly ordered or, equivalently, that $D=[n]$ for some $n\in \mathbb N$ (similar to \cite{KazanaS11}). We extend this linear ordering to a linear order on the relations of $\mathcal{D}$ via lexicographic ordering.
The \emph{Gaifman graph} of a $\sigma$-db $\mathcal D$ is the undirected graph 
$\mathcal{G}(\mathcal{D})=(V,E)$, 
with vertex set $V:=D$ and an edge between vertices $a$ and $b$ whenever $a\neq b$ and there is an 
$R\in \sigma$ and a 
tuple $(a_1,\ldots,a_{\text{ar}(R)})\in R^{\mathcal D}$ with $a,b\in\{a_1,\ldots,a_{\text{ar}(R)}\}$.
The \emph{degree} $\deg(a)$ of an element $a$ in a database $\mathcal{D}$ is the total number of tuples in all relations of $\mathcal D$ that contain $a$. We say the \emph{degree} $\deg(\mathcal{D})$ of a database $\mathcal{D}$ is the maximum degree of its elements. 
A class of databases $\mathbf{C}$ has \emph{bounded degree}, if there exists a constant $d\in\mathbb N$ such that for all $\mathcal{D} \in \mathbf{C}$, $\deg(\mathcal{D}) \leq d$. (We always assume that classes of databases are closed under isomorphism.) 
Let us remark that the $\deg(\mathcal{D})$ and the (graph-theoretic) degree of $\mathcal{G}(\mathcal{D})$ only differ by at most a constant factor (cf.\ e.\,g.~\cite{durand2007first}). Hence both measures yield the same classes of relational structures of bounded degree.
We define the \emph{tree-width} of a database $\mathcal D$ %$\tw(\mathcal D)$, 
as the the tree-width of its Gaifman graph. (See e.\,g.\ \cite{Flum:2006:PCT:1121738} for a discussion of tree-width in this context.)
A class $\mathbf{C}$ of databases has \emph{bounded tree-width}, if there exists a constant $t\in \mathbb N$ such that all databases $\mathcal{D} \in \mathbf{C}$ have tree-width at most~$t$. 
Let $\mathcal D$ be a $\sigma$-db, and $M\subseteq D$. The sub-database of  $\mathcal{D}$ \emph{induced by} $M$ is the database $\mathcal{D}[M]$ with domain $M$ and $R^{\mathcal{D}[M]}:=R^{\mathcal{D}}\cap M^{\text{ar}(R)}$ for every $R\in \sigma$.

\subsection*{Database queries.}
Let \textbf{var} be a countable infinite set of \emph{variables}, and fix a relational schema $\sigma$. 
The set $\operatorname{FO}[\sigma]$ is built from \emph{atomic formulas} of the form $x_1=x_2$ or $R(x_1, \dots, x_{\textup{ar}(R)})$, where $R \in \sigma$ and $x_1,\dots,x_{\textup{ar}(R)} \in \textbf{var}$, and is closed under  Boolean connectives ($\lnot$, $\lor$, $\land$, $\rightarrow$, $\leftrightarrow$) 
and existential and universal quantifications ($\exists, \forall$). The set $\operatorname{FO}[\{E\}]$ is the set of first-order formulas for undirected graphs.
We let $\operatorname{FO}:=\bigcup_{\sigma\text{ schema}}\operatorname{FO}[\sigma]$.
We use $\exists^{\geq m}x\,\phi$  (and $\exists^{= m}x\,\phi$, respectively) as a shortcut for the 
$\operatorname{FO}$ formula expressing that that the number of witnesses $x$ satisfying $\phi$ is at least $m$
 (exactly $m$, resp.).
 A \emph{free variable} of an $\operatorname{FO}$ formula is a variable that does not appear in the scope of a quantifier. For a tuple $\bar x$ of variables and a formula $\phi\in\operatorname{FO}$, we write $\phi(\bar{x})$ to indicate that the free variables of $\phi$ are exactly the variables in $\bar{x}$. An $\operatorname{FO}$ formula without free variables is called a \emph{sentence}.
 An \emph{$\operatorname{FO}$ query} (of arity $k\in \mathbb N$) is an $\operatorname{FO}$ formula $\phi(\bar{x})$  (with $|\bar x|=k$).
 Let $\mathcal{D}$ be a database and $\bar{a}$ be a tuple of elements of $\mathcal{D}$ of length $|\bar x |$. We write $\mathcal{D} \models \phi(\bar{a})$, if $\phi$ is true in $\mathcal{D}$ when we replace the free variables of $\phi$ with $\bar{a}$, and we say that $\bar{a}$ is an \emph{answer} for $\phi$ on $\mathcal{D}$. We let $\phi(\mathcal D):=\{\bar{a}\in D^{|\bar{a}|}\mid \mathcal{D} \models \phi(\bar{a})\}$ be the set of all answers for $\phi$ on $\mathcal D$. Two formulas $\phi(\bar{x}), \psi(\bar{x}) \in \operatorname{FO}[\sigma]$, where $|\bar{x}| =k$, are \emph{d-equivalent} (written $\phi(\bar{x}) \equiv _ d \psi(\bar{x})$) if for all $\sigma$-dbs $\mathcal{D}$ with degree at most~$d$ and all $\bar{a} \in D^k$, $\mathcal{D} \models \phi(\bar{a})$ iff  $\mathcal{D} \models \psi(\bar{a})$. The \emph{quantifier rank} of a formula $\phi$, denoted by $qr(\phi)$, is the maximum nesting depth of quantifiers that occur in $\phi$. The \emph{size of a formula} $\phi$, denoted by $\|\phi\|$, is the length of $\phi$ as a string over the alphabet $\sigma \cup \textbf{var} \cup \{\exists, \forall,\lnot,\lor,\land,\rightarrow,\leftrightarrow,=\} \cup \{,\} \cup \{(,)\}$. 
 
\subsection*{Enumeration problems.}
Let $\sigma$ be a relational schema, let $\mathbf{C}$ be a class of $\sigma$-dbs and let $\phi(\bar{x}) \in \operatorname{FO}[\sigma]$. The \emph{enumeration problem of $\phi$ over $\mathbf{C}$} denoted by $\operatorname{Enum}_{\mathbf{C}}(\phi)$ is, given a database $\mathcal{D} \in \mathbf{C}$, to output the elements of $\phi(\mathcal{D})$ one by one with no repetition. An \emph{enumeration algorithm} for the enumeration problem $\operatorname{Enum}_{\mathbf{C}}(\phi)$ with input database $\mathcal{D} \in \mathbf{C}$ proceeds in two phases, a preprocessing phase and an enumeration phase. The enumeration phase outputs all the elements of $\phi(\mathcal{D})$ with no duplicates. Furthermore, the enumeration phase has full access to the output of the preprocessing phase but can use only a constant total amount of extra memory.

The \emph{delay} of an enumeration algorithm is the maximum time between the start of the enumeration phase and the first output (or the `end of enumeration message' if there are no answers), two consecutive outputs, and the last output and the `end of enumeration message'.

%An \emph{enumeration problem} is a binary relation. Let $R$ be some enumeration problem over sets $A$ and $B$ and let $x \in A$ be an input, an \emph{answer} for $x$ is a $y \in B$ such that $(x,y) \in R$. An enumeration problem induces a computational problem: Given some input $x$, output all of its answers. An \emph{enumeration algorithm} for an enumeration problem is an algorithm that, with input $x$, proceeds in two phases, a preprocessing phase and 
%an enumeration phase. The enumeration phase outputs all the answers for $x$ with no duplicates. Furthermore, the enumeration phase has full access to the output of the preprocessing phase but can use only a constant total amount of extra memory.
%
%The \emph{delay} of an enumeration algorithm is the maximum time between 
%\begin{itemize}
%\item the start of the enumeration phase and the first output (or the `end of enumeration message' if there are no answers),
%\item two consecutive outputs, and
%\item the last output and the `end of enumeration message'.
%\end{itemize}
%
%We are interested in the following enumeration problem for $\phi (\bar{x}) \in \operatorname{FO}[\sigma]$ (with $|\bar{x}| = k$) and a class of databases $\mathbf{C}$ with bounded degree $d \in \mathbb{N}$:
%\begin{equation*}\operatorname{Enum}_{\mathbf{C}}(\phi) = \{(\mathcal{D}, \bar{a}) \mid \mathcal{D} \in \mathbf{C}\text{, } \bar{a}\text{ is a  tuple of elements of }\mathcal{D} \text{ and }\mathcal{D} \models \phi ( \bar{a}) \}.\end{equation*}

\subsection*{Neighbourhoods and Hanf normal form.}
For a $\sigma$-db $\mathcal D$ and $a,b \in D$, the \emph{distance} between $a$ and $b$ in $\mathcal D$, denoted by dist$_{\mathcal D}(a,b)$, is the length of a shortest path between $a$ and $b$ in $\mathcal{G}(\mathcal{D})$. The \emph{distance} between two tuples $\bar{a} = (a_1,\dots,a_m)$ and $\bar{b} = (b_1,\dots,b_l)$ of $\mathcal{D}$ is the min$\{\text{dist}_{\mathcal D}(a_i,b_j) \mid 1 \leq i \leq m, 1\leq j \leq l \}$. Let $r \in \mathbb{N}$. 
For a tuple $\bar{a}\in D^{|\bar a|}$, we let $N^{\mathcal D}_r(\bar{a})$ denote the set of all elements of $\mathcal{D}$ that are at distance at most $r$ from $\bar{a}$. The \emph{$r$-neighbourhood} of $\bar{a}$ in $\mathcal D$, denoted by $\mathcal{N}^{\mathcal D}_r(\bar{a})$, is the tuple $(\mathcal{D}[N_r(\bar{a})], \bar{a})$ where the elements of $\bar a$ are called \emph{centres}. We omit the superscript and
write $N_r(\bar{a})$ and $\mathcal{N}_r(\bar{a})$, if $\mathcal D$ is clear from the context.
Two $r$-neighbourhoods, $\mathcal{N}_r(\bar{a})$ and $\mathcal{N}_r(\bar{b})$, are \emph{isomorphic} (written $\mathcal{N}_r(\bar{a}) \cong  \mathcal{N}_r(\bar{b})$) if there is an isomorphism between $\mathcal{D}[N_r(\bar{a})]$ and $\mathcal{D}[N_r(\bar{b})]$ which maps $\bar{a}$ to $\bar{b}$. An $\cong$-equivalence-class of $r$-neighbourhoods with $k$ centres is called an \emph{$r$-neighbourhood type} (or \emph{$r$-type} for short) with $k$ centres. We let $T_{r}^{\sigma, d}(k)$ denote the set of all $r$-types with $k$ centres and degree at most $d$, over schema $\sigma$. Note that for fixed $d$ and $\sigma$, the cardinality 
$|T_{r}^{\sigma, d}(k)|=:\operatorname{c}(r,k)$ is a constant, only depending on $r$ and $k$. 
%(depending only on $\sigma$ and $d$). 
We say that tuple $\bar{a}\in D^{|\bar a|}$ \emph{has $r$-type $\tau$}, if $\mathcal{N}_r^{\mathcal D}(\bar{a}) \in \tau$.

%For $r\in \mathbb N$, the \emph{$r$-histogram} of a database $\mathcal{D}$,
%denoted by $\operatorname{hr}(\mathcal{D})$, is the vector with
%$\operatorname{c}(r)$ components, indexed by the $r$-types, where the
%component corresponding to type $\tau$ contains the number of elements of
%$\mathcal{D}$ of $r$-type $\tau$. We sometimes omit $r$ and write
%\emph{neighbourhood type}, if $r$ is clear from the context.

Let $r \in \mathbb{N}$ and $k \in\mathbb N_{\geq 1}$. A \emph{sphere-formula}, denoted by $\operatorname{sph}_{\tau}(\bar{x})$ (where $|\bar{x}| = k$), is an FO formula which expresses that the $r$-type of $\bar{x}$ is $\tau$, where $\tau$ is some $r$-type with $k$ centres, and $r$ is called the \emph{locality radius} of the sphere-formula. 
A \emph{Hanf-sentence} is a sentence of the form $\exists ^{\geq m} x \operatorname{sph}_{\tau}(x)$, where $\tau$ is an $r$-type with one centre, and $r$ is the \emph{locality radius} of the Hanf-sentence. An FO formula is in \emph{Hanf normal form} if it is a Boolean combination of Hanf-sentences and sphere-formulas. The \emph{Hanf locality radius} of an FO formula $\phi$ in Hanf normal form is the maximum of the locality radii of the Hanf-sentences and sphere-formulas of $\phi$. 
A well-known theorem by Hanf states that on databases of bounded degree, every FO formula can be transformed into an equivalent formula in Hanf normal form~\cite{Hanf1965}. This theorem was subsequently refined as follows.

\begin{thm}[\cite{bollig2012optimal}]\label{hnf}
	For any $\phi(\bar{x}) \in \operatorname{FO}$ and $d \in \mathbb N_{\geq 1}$, 
	there exists a 
	 $d$-equivalent formula $\psi(\bar{x})$ in Hanf normal form with the same free variables as $\phi$, and 
	  $\psi$ can be computed in time $ 2^{d^{2^{\mathcal{O}(\|\phi\|)}}}$ from $\phi$.
Furthermore, the Hanf locality radius of $\psi$ is at most $4^{qr(\phi)}$.
\end{thm}

For each FO formula $\phi$, we fix a formula $\psi$ (that is computed by Theorem \ref{hnf}) that is $d$-equivalent to $\phi$ and is in Hanf normal form. We then fix the Hanf locality radius of $\phi$ to be the Hanf locality radius of $\psi$ (and so we can then refer to the Hanf locality radius of an FO formula).

\subsection*{Local and non-local first-order queries}

We call an $\operatorname{FO}[\sigma]$ formula $\phi(\bar{x})$ (with $k$ free variables) \emph{local} if there exists some $r \in \mathbb{N}$ such that for any $\sigma$-dbs $\mathcal{D}_1$ and $\mathcal{D}_2$ and tuples $\bar{a}_1 \in D_1^k$ and $\bar{a}_2 \in D_2^k$, if $\mathcal{N}^{\mathcal D_1}_r(\bar{a}_1) \cong \mathcal{N}^{\mathcal{D}_2}_r(\bar{a}_2)$ then, $\mathcal{D}_1 \models \phi(\bar{a}_1)$ if and only if $\mathcal{D}_2 \models \phi(\bar{a}_2)$. We call $r$ the \emph{locality radius of $\phi$}. If an $\operatorname{FO}$ formula is not local we say it is \emph{non-local}. 
%Note that for any $\phi(\bar{x}) \in \operatorname{FO}$ if $\phi$ is $d$-equivalent to a boolean combination of sphere-formulas then $\phi$ is local. 
We highlight that this notion of locality differs from that of Hanf locality and Gaifman locality of FO and should not be confused.  
%\polly{  hanf-localty radius is locality radius give proof}

\begin{proviso}
For the rest of the paper, we fix a schema $\sigma$ and numbers $d,t \in \mathbb{N}$ with $d \geq 2$. From now on, all databases are $\sigma$-dbs and have degree at most $d$, unless stated otherwise. We use $\mathbf{G}_d$ to denote the class of all graphs with degree at most $d$, $\mathbf{C}_d$ to denote the class of all $\sigma$-dbs with degree at most $d$, $\mathbf{C}_d^t$ to denote the class of all $\sigma$-dbs with degree at most $d$ and tree-width at most $t$ and finally we use $\mathbf{C}$ to denote a class of $\sigma$-dbs with degree at most $d$.
\end{proviso}

\subsection*{Property Testing.}\label{Prelim PT}
First, we note that we only use methods from property testing from Section~\ref{sec: non local queries} onwards.
 We use the model of property testing for bounded degree databases introduced in~\cite{adler2018property}, which is a straightforward extension of the model for bounded degree graphs~\cite{goldreich2002property}. 
Property testing algorithms do not have access to the whole input database. Instead, they are given access via an \emph{oracle}. Let $\mathcal{D}$ be an input $\sigma$-db on $n$ elements. A property testing algorithm receives the number $n$ as input, and it can make \emph{oracle queries}\footnote{Note that an oracle query is not a database query.} of the form $(R,i,j)$, where $R \in \sigma$, $i \leq n$ and $j \leq \text{deg}(\mathcal{D})$. The answer to $(R,i,j)$ is the $j^{\text{th}}$ tuple in $R^{\mathcal{D}}$ containing the $i^{\text{th}}$ element\footnote{According to the assumed linear order on $D$.} of $D$ (if such a tuple does not exist then it returns $\bot$). We assume oracle queries are answered in constant time. 

Let $\mathcal{D},\mathcal{D'}$ be two $\sigma$-dbs, both having $n$ elements. The \emph{distance} between $\mathcal{D}$ and $\mathcal{D'}$, denoted by dist$(\mathcal{D}, \mathcal{D'})$, is the minimum number of tuples that have to be inserted or removed from relations of $\mathcal{D}$ and $\mathcal{D'}$ to make $\mathcal{D}$ and $\mathcal{D'}$ isomorphic. For $\epsilon \in [0,1]$, we say $\mathcal{D}$ and $\mathcal{D'}$ are \emph{$\epsilon$-close} if dist$(\mathcal{D}, \mathcal{D'}) \leq \epsilon d n$, and are \emph{$\epsilon$-far} otherwise. A \emph{property} is simply a class of databases. Note that every $\operatorname{FO}$ sentence $\phi$ defines a property $\mathbf{P}_{\phi}=\{\mathcal D\mid \mathcal D \models \phi\}$. We call $\mathbf{P}_{\phi}\cap \mathbf{C}$ the property \emph{defined by $\phi$ on $\mathbf{C}$}.
A $\sigma$-db $\mathcal{D}$ is \emph{$\epsilon$-close} to a property $\mathbf{P}$ if there exists a database $\mathcal{D'} \in \mathbf{P}$ that is $\epsilon$-close to $\mathcal{D}$, otherwise $\mathcal{D}$ is \emph{$\epsilon$-far} from $\mathbf{P}$.

Let $\mathbf{P} \subseteq \mathbf{C}$ be a property and $\epsilon \in (0,1]$ be the proximity parameter. An \emph{$\epsilon$-tester} for $\mathbf{P}$ on $\mathbf{C}$ is a probabilistic algorithm which is given oracle access to a $\sigma$-db $\mathcal{D} \in \mathbf{C}$ and it is given $n:=|D|$ as auxiliary input. The algorithm does the following.
\begin{enumerate}
\item If $\mathcal{D} \in \mathbf{P}$, then the tester accepts with probability at least ${2}/{3}$.
\item If $\mathcal{D}$ is $\epsilon$-far from $\mathbf{P}$, then the tester rejects with probability at least ${2}/{3}$.
\end{enumerate}
The \emph{query complexity} of a tester is the maximum number of oracle queries made. 
A tester has \emph{constant} query complexity, if the query complexity does not depend on
the size of the input database.
We say a property $\mathbf{P} \subseteq \mathbf{C}$ is \emph{uniformly testable} in time $f(n)$ on $\mathbf{C}$, if for every $\epsilon \in (0,1]$ there exists an $\epsilon$-tester for $\mathbf{P}$ on $\mathbf{C}$ which has constant query complexity and whose running time on databases on $n$ elements is $f(n)$. Note that this tester must work for all $n$. We give an example below, which is also the basis of our running example. 
%These are inspired by biological networks.

\begin{exa}\label{tester example}
%	Let $\mathbf{C}$ be the class of simple graphs with bounded degree $d \in \mathbb{N}$. 
	On the class $\mathbf G_d$, consider the isomorphism types
	$\tau_2$ and $\tau_4$ of the $2$-neighbourhoods $(N_2, (c_1,c_2))$ and $(N_4, (c_1))$ where $N_2$ and $N_4$ are the graphs shown in Figure~\ref{fig:taus} with centres $(c_1,c_2)$ and $(c_1)$. 
\begin{figure}
\begin{center}
\begin{tikzpicture}
[scale=.5,auto=left,every node/.style={circle,fill=black!,scale=.4}]

  \node[fill=gray,scale=1.5, label=right:\LARGE $c_1$] (n1) at (0,0) {};
  \node (n2) at (-1,1)  {};
  \node (n3) at (1,1)  {};
 \node[fill=gray,scale=1.5, label=right:\LARGE $c_2$] (n4) at (0,-1)  {};
  \node (n5) at (-2,1)  {};
   \node (n6) at (-1,2)  {};
     \node (n7) at (2,1) {};
  \node (n8) at (1,2)  {};

  \foreach \from/\to in {n1/n2,n1/n3,n1/n4,n2/n5,n2/n6,n3/n7,n3/n8}
    \draw (\from) -- (\to);
    
    \node [below=2cm,align=flush center,style={fill=none,scale=1.5},scale=1.5] at (n1)
        {
            $N_1$
        };

 \node[fill=gray,scale=1.5, label=right:\LARGE $c_1$] (n11) at (6,0) {};
  \node (n12) at (5,1)  {};
  \node (n13) at (7,1)  {};
 \node[fill=gray,scale=1.5, label=right:\LARGE $c_2$] (n14) at (6,-1)  {};
  \node (n15) at (4,1)  {};
   \node (n16) at (5,2)  {};
     \node (n17) at (8,1) {};
  \node (n18) at (7,2)  {};

  \foreach \from/\to in {n11/n12,n11/n13,n11/n14,n12/n15,n12/n16,n13/n17,n13/n18,n16/n15}
    \draw (\from) -- (\to);
    
    \node [below=2cm,align=flush center,style={fill=none,scale=1.5},scale=1.5] at (n11)
        {
            $N_2$
        };

 \node[fill=gray,scale=1.5, label=right:\LARGE $c_1$] (n21) at (12,0) {};
  \node (n22) at (11,1)  {};
  \node (n23) at (13,1)  {};
 \node[fill=gray,scale=1.5, label=right:\LARGE $c_2$] (n24) at (12,-1)  {};
  \node (n25) at (10,1)  {};
   \node (n26) at (11,2)  {};
     \node (n27) at (14,1) {};
  \node (n28) at (13,2)  {};

  \foreach \from/\to in {n21/n22,n21/n23,n21/n24,n22/n25,n22/n26,n23/n27,n23/n28,n25/n26,n27/n28}
    \draw (\from) -- (\to);
    
    \node [below=2cm,align=flush center,style={fill=none,scale=1.5},scale=1.5] at (n21)
        {
            $N_3$
        };

  \node[fill=gray,scale=1.5, label=right:\LARGE $c_1$] (n29) at (18,0) {};
  \node (n30) at (17,1)  {};
  \node (n31) at (19,1)  {};
 \node (n32) at (18,-1)  {};
  \node (n33) at (16,1)  {};
   \node (n34) at (17,2)  {};
     \node (n35) at (20,1) {};
  \node (n36) at (19,2)  {};

  \foreach \from/\to in {n29/n30,n29/n31,n29/n32,n30/n33,n30/n34,n31/n35,n31/n36}
    \draw (\from) -- (\to);
    
    \node [below=2cm,align=flush center,style={fill=none,scale=1.5},scale=1.5] at (n29)
        {
            $N_4$
        };
\end{tikzpicture}
\end{center}
\caption{The four $2$-types of Examples~\ref{tester example} and \ref{running example}.
The vertices labelled `$c_1$' and `$c_2$' are the centres.}
\label{fig:taus}
\end{figure}
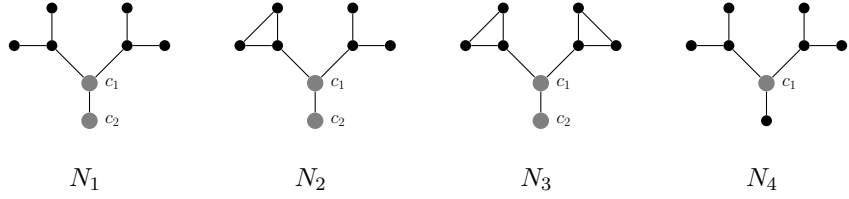
	Let $\phi$ be the $\operatorname{FO}[\{E\}]$-formula $\phi = \exists x \exists y \operatorname{sph}_{\tau_2}(x,y) \land \lnot \exists z\operatorname{sph}_{\tau_4}(z).$
	Consider the property $\mathbf{P}:=\mathbf{P}_{\phi}$. We show that on the class $\mathbf{G}_d$, $\mathbf{P}\cap \mathbf{G}_d$ is uniformly testable with constant time. For this, let $\epsilon \in (0,1]$. Given oracle access to a graph $\mathcal{G} \in \mathbf{G}_d$ and $|V(\mathcal{G})|=n$ as an input, the $\epsilon$-tester proceeds as follows:
\begin{enumerate}
\item If $n < 24d^3/\epsilon$, do a full check of $\mathcal{G}$ and decide if $\mathcal{G} \in \mathbf{P}$. 
\item Otherwise uniformly and independently sample $\alpha = \log_{1-\epsilon d/3}1/3$ vertices from $[n]$.
\item For each sampled vertex, compute its $2$-neighbourhood.
\item If a vertex is found with $2$-type $\tau_4$ then the tester rejects. Otherwise it accepts.
\end{enumerate}
\begin{claim}\label{claim: pt example}
The above $\epsilon$-tester accepts with probability at least $2/3$ if $\mathcal{G} \in \mathbf{P}$ and rejects with probability at least $2/3$ if $\mathcal{G}$ is $\epsilon$-far from $\mathbf{P}$. Furthermore, the $\epsilon$-tester has constant query complexity and runs in constant time.
\end{claim}
\begin{claimproof}
	Note that in $\tau_4$, every vertex has $1$ or $3$ neighbours.
For showing correctness, first assume $\mathcal{G} \in \mathbf{P}$. Then the tester will always accept as there exists no vertex with 2-type $\tau_4$.

	Now assume $\mathcal{G}$ is $\epsilon$-far from $\mathbf{P}$. Then at least $\epsilon d n$ edge modifications are necessary to make $\mathcal{G}$ isomorphic to a graph in $\mathbf{P}$. If $n<24d^3/\epsilon$ then the tester will reject so assume otherwise. Inserting a copy of $\tau_2$ requires at most $8(d+1)$ modifications (pick 8 vertices and remove all incident edges then add the 8 edges to make an isolated copy of $\tau_2$). Removing an edge $uv$ from $\mathcal{G}$ will change the $2$-type of any vertex in the set $N_2^{\mathcal{G}}(u) \cap N_2^{\mathcal{G}}(v)$. Lemma 3.2 (a) of \cite{berkholz2018answering} states that $|N_2^{\mathcal{G}}(u)| \leq d^{2+1}$ and $|N_2^{\mathcal{G}}(v)| \leq d^{2+1}$. Therefore, $|N_2^{\mathcal{G}}(u) \cap N_2^{\mathcal{G}}(v)| \leq d^3$ and inserting a copy of $\tau_2$ could add at most $8d^4$ many copies of $\tau_4$.  After inserting a copy of $\tau_2$ we need to remove all copies of $\tau_4$. Let $v \in V(\mathcal{G)}$ be a vertex with 2-type $\tau_4$. Let $u$ be the neighbour of $v$ with degree 1. If we remove the edge $uv \in E(\mathcal{G})$, $v$'s $2$-type is no longer $\tau_4$. Note that $v$ has exactly $2$ neighbours and $u$ has $0$ neighbours in $\mathcal G\setminus uv$.
	Moreover, we claim that by removing $uv$, we have introduced no new vertices with $2$-type $\tau_4$.
	To see this, observe that deleting $uv$ will only affect the $2$-types of vertices in 
	$N^{\mathcal{G}}_1(v)$. But each vertex $x\in N^{\mathcal{G}}_1(v)$ will have a vertex with exactly two neighbours in its $2$-neighbourhood in $\mathcal G\setminus uv$. Hence
	the new $2$-type of $x$ is not $\tau_4$.	
	This shows that there are at least $ \epsilon d n - 8(d+1) - 8d^4$ vertices with $2$-type $\tau_4$. As $n \geq 24d^3/\epsilon$ and $d \geq 2$, then $8(d+1)  \leq 8d^4 \leq \epsilon d n / 3$. The probability that we sample a vertex with $2$-type $\tau_4$ is therefore at least $\epsilon d n/3n = \epsilon d/3$. Hence the probability that none of the $\alpha$ sampled vertices have $2$-type $\tau_4$ is at most $(1- \epsilon d/3)^\alpha = 1/3$. Therefore with probability at least $2/3$ the tester rejects.

	For the running time, if $n<24d^3/\epsilon$ then we can do a full check of the input graph in time only dependent on $d$ and $\epsilon$. Otherwise, note that the tester samples only a constant number of vertices in (2), and for 
	each of the sampled vertices, the tester needs to make a constant number of oracle queries only to calculate its $2$-neighbourhood in (3), because the degree is bounded. 
	Therefore the tester has constant query complexity and constant running time.
	\end{claimproof}
\end{exa}

Adler and Harwath showed that, on the class of all databases with bounded degree and tree-width, every property definable in monadic second-order logic with counting (CMSO) is uniformly testable in polylogarithmic running time~\cite{adler2018property}. (Where a function is \emph{polylogarithmic} in $n$, if it is a polynomial in $\log n$.)
The logic CMSO is an extension of FO and we, therefore, get the following result which will be used as a subroutine in Section~\ref{sec: non local queries}.

\begin{thm}[\cite{adler2018property}]\label{FO testability} 
 Each property $\mathbf{P} \subseteq \mathbf{C}_d^t$ definable in $\operatorname{FO}$ is uniformly testable on $\mathbf{C}_d^t$ in polylogarithmic running time.
\end{thm}

\subsection*{Model of Computation.}
We use Random Access Machines (RAMs) and a uniform cost measure when analysing our algorithms, i.\,e.\ we assume all basic arithmetic operations including random sampling can be done in constant time, regardless of the size of the numbers involved. 
We assume that if we initialise an array, all entries are set to 0 and this can be done in constant time for any length or dimension array. This is achieved by using the lazy array initialisation technique (cf.\ e.g.~\cite{Moret:1991:APN:102912}) where entries are only actually stored when they are first needed. We use one-based indexing for arrays. Let $\mathbf{A}$ be a $1$-dimensional array. We assume that for a number $a \in \mathbb{N}_{\geq 1}$, the entry $\mathbf{A}[a]$ can be accessed in constant time.
% Let $k \in \mathbb{N}_{\geq 1}$ and $\mathbf{A}$ be a $k$-dimensional array. We assume that for a tuple 
%$(a_1,a_2,\dots,a_k) \in  \mathbb{N}_{\geq 1}^k$, the entry  $\mathbf{A}[a_1,a_2,\dots,a_k]$
%at position $(a_1,\dots,a_k)$ can be accessed in constant time.
%The address of the entry $\mathbf{A}[a_1,\dots,a_k]$ can be found by a formula that requires $k$ multiplications and $k$ additions and hence in the RAM-model the entry $\mathbf{A}[a_1,\dots,a_k]$ can be accessed in constant time.

\section{Properties of first-order queries on bounded degree}\label{sec: local and non local queries}

In this section, we shall give some useful normal forms of FO queries. We shall then give a characterisation and some results for local FO queries.  

\subsection{General first-order queries}
We make use of the following lemma to simplify Boolean combinations of sphere-formulas. We shall use this result to show we can write FO queries in a special type of Hanf normal form that groups the Hanf-sentences and the sphere-formulas in a convenient way. 

\begin{lem}[\cite{berkholz2018answering}]\label{d-equivalent disjunction}
Let $r,k,d\in \mathbb N$ with $k \geq 1$, $d \geq 2$ and let $\sigma$ be a schema. For every Boolean combination $\phi(\bar{x})$ of sphere-formulas of degree at most $d$ and radius at most $r$, there exists an $I \subseteq T_{r}^{\sigma, d}(k)$ such that $\phi(\bar{x})$ is $d$-equivalent to $\bigvee_{\tau \in I}\operatorname{sph}_{\tau}(\bar{x})$.

Furthermore, given $\phi(\bar{x})$, the set $I$ can be computed in time $ \poly(\|\phi\|)\cdot 2^{(kd^{r+1})^{\mathcal{O}(\|\sigma\|)}}.$
\end{lem}

In the following lemma, we show that we can write any FO query as a disjunction of conjunctions of a sphere-formula and a boolean combination of Hanf-sentences. This normal form will be used in Lemma \ref{lemma: computing rel r-types}.

%where we use this form to help us compute a set of $r$-types such that with high probability, 
%
%for any $\bar{a} \in D^k$,
%\begin{enumerate}
%\item if $\bar{a} \in \phi(\mathcal{D})$, then the $r$-type of $\bar{a}$ in $\mathcal{D}$ is in $T$, and
%\item if $\bar{a} \in D^k \setminus \phi(\mathcal{D},\mathbf{C}_d^t, \epsilon )$, then the $r$-type of $\bar{a}$ in $\mathcal{D}$ is not in $T$.
%\end{enumerate} 

\begin{lem}\label{normal-form lemma}
Let $\phi(\bar{x}) \in \operatorname{FO}$ and $|\bar{x}|=k$. Let $r$ be the Hanf locality radius of $\phi$. For every $d \in \mathbb{N}$ with $d  \geq 2$ there exists a computable, $d$-equivalent formula to $\phi$ of the form
\begin{equation}\label{normal-form}
	\chi(\bar{x}) =\bigvee_{i\in [m]}\Big(\operatorname{sph}_{\tau_i}(\bar{x}) \land \psi^s_i \Big)
\end{equation}
for some $m \in \mathbb{N}$, where for all $i \in [m]$, $\tau_i$ is an $r$-type with $k$ centres and $\psi^s_i$ is a conjunction of Hanf-sentences and negated Hanf-sentences. For each $\phi(\bar{x}) \in \operatorname{FO}$, we fix such a $d$-equivalent formula to $\phi$ (so we can refer to \emph{the} $d$-equivalent formula of $\phi$ in the form (\ref{normal-form})).
\end{lem}
\begin{proof}
From $\phi$ we can construct a formula in the required form as follows. Firstly, by Theorem \ref{hnf} we construct a $d$-equivalent formula in Hanf normal form. Next, we write the resulting formula in DNF to obtain a formula of the form
\begin{equation*}
	\chi(\bar{x})' =\bigvee_{i\in[l]}\Big(\psi_i^f(\bar{x}) \land \psi^s_i \Big)
\end{equation*}
for some $l \in \mathbb{N}$, where for $i \in [l]$, $\psi_i^f(\bar{x})$ is a conjunction of sphere-formulas and $\psi^s_i$ is a conjunction of Hanf-sentences and negated Hanf-sentences. Then, by Lemma \ref{d-equivalent disjunction}, we can replace each $\psi_i^f(\bar{x})$ with a $d$-equivalent formula $ \bigvee_{t \in \lambda_i}\operatorname{sph}_{t}(\bar{x})$ where $\lambda_i$ is a set of $r$-types with $k$ centres. Finally, we replace each $\bigvee_{t \in \lambda_i}\operatorname{sph}_{t}(\bar{x}) \land \psi^s_i$ with $ \bigvee_{t_i \in \lambda_i}( \operatorname{sph}_{t_i}(\bar{x}) \land \psi^s_i)$. The resulting formula is in the required form.
\end{proof}

In Theorems \ref{theorem: strengthened local main} and \ref{theorem: strengthened non local main}, we reduce the minimum size of the answer set required to enumerate all answers to the query $\phi$ in our approximate enumeration algorithms. We show we only actually require an answer set of size $\gamma n ^c$, where $c:=\operatorname{conn}(\phi,d)$ is the maximum number of connected components in the $r$-neighbourhood (where $r$ is the Hanf-locality radius of $\phi$) of an answer to $\phi$. We define $\operatorname{conn}(\phi,d)$ below.

\begin{defi}[$\operatorname{conn}(\phi,d)$] Let $\phi(\bar{x}) \in \operatorname{FO}[\sigma]$ where $|\bar{x}| =k$ and let $\chi(\bar{x})$ be the formula in the form (\ref{normal-form}) of Lemma \ref{normal-form lemma} that is $d$-equivalent to $\phi$. We define \emph{$\operatorname{conn}(\phi,d)$} as the maximum number of connected components of the neighbourhood types that appear in the sphere-formulas of $\chi$.
Note that $\operatorname{conn}(\phi,d) \leq k$.
\end{defi}

Recall that we fix a formula $\chi$ in the form (\ref{normal-form}) of Lemma \ref{normal-form lemma} for each FO formula $\phi$, and hence $\operatorname{conn}(\phi,d)$ is well defined.

\subsection{Local first-order queries}
We shall start by showing that for any local FO query $\phi$ we can compute a set of $r$-types $T$ (where $r$ is the locality radius) such that for any $\sigma$-db $\mathcal{D}$ and tuple $\bar{a}$, $\bar{a}$ is an answer to $\phi$ on $\mathcal{D}$ if and only if the $r$-type of $\bar{a}$ is in $T$. 

\begin{lem}\label{lemma: local query rel r-types}
There is an algorithm that, given a local query $\phi(\bar{x})\in \operatorname{FO}[\sigma]$ with $k$ free variables and given the locality radius $r$ of $\phi$, computes a set of $r$-types $T$ with $k$ centres such that for any $\sigma$-db $\mathcal{D}$ and tuple $\bar{a} \in D^k$, $\bar{a} \in \phi(\mathcal{D})$ if and only if the $r$-type of $\bar{a}$ in $\mathcal{D}$ is in $T$.
\end{lem}
\begin{proof}
 Let $T$ be an empty list. For each $r$-type $\tau$ with $k$ centres we do the following. Let $\mathcal{D}_{\tau}$ be the fixed representative $\sigma$-db of $\tau$ where $\bar{c}$ is the centre tuple, then if $\mathcal{D}_{\tau} \models \phi(\bar{c})$ add $\tau$ to $T$. Then since $\phi$ is local and $r$ is the locality radius of $\phi$, for every $\sigma$-db $\mathcal{D}$ and tuple $\bar{a} \in D^k$, $\mathcal{D} \models \phi(\bar{a})$ if and only if the $r$-type of $\bar{a}$ in $\mathcal{D}$ is in $T$.
\end{proof}

Using the previous lemma we shall show that for any local FO query, $\sigma$-db $\mathcal{D}$ and tuple $\bar{a}$ from $\mathcal{D}$ it can be decided in constant time whether $\bar{a}$ is an answer to $\phi$ on $\mathcal{D}$. We will use this when approximately enumerating answers to local FO queries.

\begin{lem}\label{lemma: local membership testing} Let $\phi(\bar{x})\in \operatorname{FO}[\sigma]$ be a local query with $k$ free variables. There is an algorithm that, given a $\sigma$-db $\mathcal{D}$ and a tuple $\bar{a} \in D^k$, decides whether $\bar{a} \in \phi(\mathcal{D})$ in constant time.
\end{lem}
\begin{proof}
%First note that for any constant $r$, the $r$-neighbourhood type of the tuple $\bar{a}$ in $\mathcal{D}$ can be computed in constant time.

Let $r$ be the locality radius of $\phi$. First let us compute the set of $r$-types $T$ as in Lemma \ref{lemma: local query rel r-types}. We shall then compute the $r$-type $\tau$ of $\bar{a}$ in $\mathcal{D}$. By Lemma \ref{lemma: local query rel r-types} if $\tau \in T$ then $\bar{a} \in \phi(\mathcal{D})$ and if $\tau \not\in T$ then $\bar{a} \not\in \phi(\mathcal{D})$.

Since $r$ does not depend on $\mathcal{D}$, the $r$-type of $\bar{a}$ in $\mathcal{D}$ can be computed in constant time. Furthermore, computing the set $T$ does not depend on $\mathcal{D}$, and hence it can be decided in constant time whether $\bar{a} \in \phi(\mathcal{D})$.
%
%Let $\chi$ be the formula that is $d$-equivalent to $\phi$ in the form (\ref{local normal-form}) of Lemma \ref{lemma: hanf locality radius and locality radius}. Note $\chi$ can be computed in time independent of $|D|$. Let $r$ be the Hanf locality radius of $\chi$ and let $\tau$ be the $r$-type of $\bar{a}$ in $\mathcal{D}$ (which can be computed in constant time). If there exists a sphere formula $\operatorname{sph}_{\tau}(\bar{x})$ in $\chi$ then clearly $\mathcal{D} \models \chi(\bar{a})$ and hence $\bar{a} \in \phi(\mathcal{D})$. If there does not exist a sphere formula $\operatorname{sph}_{\tau}(\bar{x})$ in $\chi$ then clearly $\mathcal{D} \not\models \chi(\bar{a})$ and hence $\bar{a} \not\in \phi(\mathcal{D})$. Therefore it can be decided in constant time whether $\bar{a} \in \phi(\mathcal{D})$ or not.
\end{proof}

We shall finish this section with the following characterisation of local FO queries. We do not make use of this characterisation but we include it to aid intuition. The proof of the observation is straightforward but we shall give it for completeness.

\begin{obs}\label{lemma: hanf locality radius and locality radius} 
%Let $\sigma$ be a schema and let $\phi(\bar{x})\in \operatorname{FO}[\sigma]$ be a local query. Let $r$ be the Hanf locality radius of $\phi$. For every $d \in \mathbb{N}$ with $d  \geq 2$ there exists a, $d$-equivalent formula to $\phi$ which is of the form
% \begin{equation}\label{local normal-form}
%	\chi(\bar{x}) =\bigvee_{i\in [m]}\operatorname{sph}_{\tau_i}(\bar{x}) 
%\end{equation}
%where each $\tau_i$ is an $r$-type with $|\bar{x}|$ centres.
%Moreover, $\chi$ is computable from $\phi$.
Let $\phi(\bar{x})\in \operatorname{FO}[\sigma]$. Then $\phi$ is local if and only if $\phi$ is $d$-equivalent to a boolean combination of sphere-formulas.

Furthermore, for any local FO query $\phi$, the locality radius of $\phi$ is equal to the Hanf locality radius of $\phi$. Therefore, since the Hanf locality radius of an FO query is computable by Theorem \ref{hnf}, the locality radius of a local FO query is also computable.
\end{obs}
\begin{proof}
We will give a proof of the first part of the observation only.
 Let $|\bar{x}| =k$.
First let us assume that $\phi$ is $d$-equivalent to a FO formula $\chi$ that is a boolean combination of sphere-formulas. Let $r$ be the Hanf locality radius of $\chi$. Then since $\chi$ contains no Hanf-sentences, for any $\sigma$-dbs $\mathcal{D}_1$ and $\mathcal{D}_2$ and tuples $\bar{a}_1 \in D_1^k$ and $\bar{a}_2 \in D_2^k$, if $\mathcal{N}^{\mathcal D_1}_r(\bar{a}_1) \cong \mathcal{N}^{\mathcal{D}_2}_r(\bar{a}_2)$ then, $\mathcal{D}_1 \models \phi(\bar{a}_1)$ if and only if $\mathcal{D}_2 \models \phi(\bar{a}_2)$. Hence $\phi$ is local and $r$ is the locality radius of $\phi$.
 
Now let us assume that $\phi$ is local. Let $T$ be the set of $r$-types as constructed in Lemma \ref{lemma: local query rel r-types}. Therefore $\phi$ is $d$-equivalent to the formula $\bigvee_{\tau \in T} \operatorname{sph}_{\tau}(\bar{x})$ which is in the required form.
\end{proof}

\section{Enumerating Answers to Local First-Order Queries}\label{sec: local queries}

Assume $q$ is a local FO query with $k$ free variables and $\mathcal{D}$ is a $\sigma$-db, such that the set $q(\mathcal{D})$ is larger than a fixed proportion of all possible $k$-tuples, i.\,e.\ $|q(\mathcal{D})| \geq \mu |D|^k$ for some fixed $\mu \in (0,1)$. It is easy to construct an algorithm that enumerates the set $q(\mathcal{D})$ with amortized constant delay, i.\,e.~the average delay between any two outputs is constant. For each tuple $\bar{a} \in D^k$ (processed in, say, lexicographical order), the algorithm tests if $\bar{a}$ is in $q(\mathcal{D})$ (which can be done in constant time by Lemma \ref{lemma: local membership testing} as $q$ is local) and outputs $\bar{a}$ if $\bar{a} \in q(\mathcal{D})$. Since we are assuming that $|q(\mathcal{D})|$ is larger than a fixed proportion of all possible tuples, the overall running time of the algorithm is $\mathcal{O}(|D|^k)$ and hence the algorithm has constant amortized delay. In this section we prove that we can de-amortize this algorithm using random sampling.

We begin this section by showing that there exists a randomised algorithm which does the following. The input is a set $V$ which is partitioned into two sets $V_1$ and $V_2$. We assume that the algorithm can test in constant time if a given element from $V$ is in $V_1$ or $V_2$. After a constant time preprocessing phase, the algorithm enumerates a set $S$ of elements with $S \subseteq V_1$, with constant delay. Furthermore, we show that if $|V_1|$ is large enough then with high probability $S=V_1$. We then use this result to prove our main theorem of this section (Theorem \ref{thm: local queries enumeration}) on the approximate enumeration of the answers to a local query. In Theorem~\ref{theorem: strengthened local main} we show that the relative size of the answer set can be reduced whilst still guaranteeing that with high probability we enumerate all answers to the query.

\begin{lem}\label{lemma: partitioned set enumeration} Fix $\mu \in (0,1)$ and $\delta \in (0,1)$. There exists a randomised algorithm which does the following. The input is a set $V$ which is partitioned into two sets $V_1$ and $V_2$. We assume that the algorithm is given access to the size of $V$ and can decide in constant time whether a given element from $V$ is in $V_1$ or $V_2$. The algorithm outputs a set $S \subseteq V_1$ such that if $|V_1| \geq \mu |V|$ then, with probability at least $\delta$, $S=V_1$.

The algorithm has constant preprocessing time and enumerates $S$ with no duplicates and constant delay between any two consecutive outputs. 
\end{lem}

\begin{proof}
Let $|V| =n$ and let us assume that $V$ comes with a linear order over its elements, or equivalently that $V=[n]$. If $V$ does not come with a linear order over its elements then we use the linear order induced by the encoding of $V$. Let $q=\text{min}((1-\mu(1-\mu))^2, ({1-\delta})^2/{9})$.
The preprocessing phase proceeds as follows:
\begin{enumerate}
\item Initialise an array $\mathbf{B}$ of length $n$. The array $\mathbf{B}$ contains one entry for each element in $[n]$ and it is used to record sampled elements. For an element $a \in [n]$, the entry $\mathbf{B}[a]$ is 1 if $a$ has previously been sampled and it is 0 otherwise.
\item Initialise an empty queue $\mathbf{Q}$, to store tuples to be enumerated.
%\item Sample $\alpha = \lceil \log_{1-\gamma(1-\gamma)}q \rceil$ many elements uniformly and independently from $[n]$. 
%\item For each sampled element $a$, if $\mathbf{B}[a] = 1$, skip this element. Otherwise, set $\mathbf{B}[a]=1$ and if $a$ is in $V_1$ add $a$ to $\mathbf{Q}$.
\end{enumerate}

 As discussed in Section \ref{section prelim} an array of any size can be initialised in constant time using lazy initialisation and hence the preprocessing phase runs in constant time.

	Moving on to the enumeration phase, between each output the algorithm will sample a constant number of elements as well as going through a constant number of the elements in $[n]$ in order. The enumeration phase proceeds as follows:
\begin{enumerate}

\item Sample $\alpha = \lceil \log_{1-\mu(1-\mu)}q \rceil $ many elements uniformly and independently from $[n]$ and let \emph{t} be a list of these elements.
\item Add the next $\lceil {1}/{\mu^2} \rceil$ elements from $[n]$ to \emph{t}. If there are less than $\lceil {1}/{\mu^2} \rceil$ elements remaining just add all the remaining elements to \emph{t}.
\item For each element $a$ in \emph{t}, if $\mathbf{B}[a] = 1$, skip this element. Otherwise, set $\mathbf{B}[a]=1$ and if $a$ is in $V_1$ add $a$ to $\mathbf{Q}$.
\item If $\mathbf{Q}\neq\emptyset$, output the next element from $\mathbf{Q}$; stop otherwise.
\item Repeat Steps 1-4 until there is no element to output in Step 4.
\end{enumerate}

In Steps 1 and 2 a list of elements is created which is of constant size. For each element in this list, in Step 3, the algorithm can check whether it is in $V_1$ in constant time and the arrays $\mathbf{Q}$ and $\mathbf{B}$ can be read and updated in constant time. 
Hence, each enumeration step can be done in constant time. 
This concludes the analysis of the running time. We now prove correctness. 

Clearly, no duplicates will be enumerated due to the use of the array $\mathbf{B}$ and the only elements enumerated are those that are in $V_1$.
Let $S$ be the set of elements that are enumerated. We need to show that with probability at least ${2}/{3}$ if $|V_1| \geq \mu |V|$, then $S = V_1$. In each enumeration step we take the next $\lceil{1}/{\mu ^2}\rceil$ elements from $[n]$. Assuming $|V_1| \geq \mu |V|$, after $ \lceil n \cdot \mu^2 \rceil \leq \lceil \mu |V_1| \rceil$ enumeration steps the algorithm will have checked every element in $[n]$ and therefore $S=V_1$. Let us find a bound on the probability that we do at least $\lceil \mu |V_1| \rceil $ enumeration steps.

	\begin{claim}\label{probability bound}
For all $q \in [0,1)$ and $m \in \mathbb{N}_{\geq 1}$, $\prod_{i=1}^{m}(1-q^{\frac{i+1}{2}}) \geq 1-3q^{\frac{1}{2}}.$
	\end{claim}
		\begin{claimproof}
	First let us prove that 
\[\prod_{i=1}^{m}(1-q^{\frac{i+1}{2}}) \geq1-q^{\frac{1}{2}}-q-q^{\frac{3}{2}} +q^{\frac{m+2}{2}}\]
by induction on $m$.

 For the base case, let $m=1$, then
\[\prod_{i=1}^{1}(1-q^{\frac{i+1}{2}}) = 1-q \geq 1-q^{\frac{1}{2}}-q-q^{\frac{3}{2}} +q^{\frac{3}{2}}\]
as required.

Now for the inductive step. Let us assume the claim is true for $m$ and we shall show the claim is true for $m+1$. We have
\begin{align*} \prod_{i=1}^{m+1}(1-q^{\frac{i+1}{2}}) & = \Big(\prod_{i=1}^{m}(1-q^{\frac{i+1}{2}})\Big)\cdot (1-q^{\frac{m+2}{2}}) & \geq (1-q^{\frac{1}{2}}-q-q^{\frac{3}{2}} +q^{\frac{m+2}{2}})(1-q^{\frac{m+2}{2}})\end{align*}
by the inductive hypothesis.
\begin{align*}(1-q^{\frac{1}{2}}-q-q^{\frac{3}{2}} +q^{\frac{m+2}{2}})(1-q^{\frac{m+2}{2}})  &= 1-q^{\frac{1}{2}}-q-q^{\frac{3}{2}} +q^{\frac{m+3}{2}}+q^{\frac{m+4}{2}}+q^{\frac{m+5}{2}}-q^{m+2} \\ & \geq 1-q^{\frac{1}{2}}-q-q^{\frac{3}{2}} +q^{\frac{m+3}{2}},\end{align*}
as $q^{(m+4)/2}+q^{(m+5)/2}-q^{m+2} \geq 0$.

Therefore,  
\[\prod_{i=1}^{m}(1-q^{\frac{i+1}{2}}) \geq 1-q^{\frac{1}{2}}-q-q^{\frac{3}{2}} +q^{\frac{m+2}{2}} \geq1-3q^{\frac{1}{2}}\] as required
\end{claimproof}
	\begin{claim} \label{printing probability}
	Assume that $|V_1| \geq \mu n$.  The probability that at least $\lceil \mu|V_1| \rceil$ distinct elements from $V_1$ are enumerated is at least $1-3q^{\frac{1}{2}}$.
\end{claim} 
\begin{claimproof}
We shall start by showing that for $j \in \mathbb{N}$, where $1 \leq j \leq \lceil \mu|V_1| \rceil$, the probability that at least $j$ distinct elements from $V_1$ are enumerated is at least $\prod_{i=1}^j(1-q^{(i+1)/{2}})$.

 We shall prove this by induction on $j$. For the base case, let $j=1$. If an element from $V_1$ is sampled in the first enumeration step, then at least one element from $V_1$ will be enumerated. An element that is in $V_1$ is sampled with probability
 \[\frac{|V_1|}{n} \geq \frac{\mu n}{n} = \mu  \geq \mu(1-\mu).\]
The probability that out of the $\alpha$ elements sampled in the first enumeration step there is none from $V_1$ is at most
 $  (1-\mu(1-\mu))^{\alpha} \leq q$ as $\alpha = \lceil \log_{1-\mu(1-\mu)}q \rceil \geq \log_{1-\mu(1-\mu)}q $.
Therefore with probability at least $1-q$ at least one element from $|V_1|$ is enumerated and hence we have proved the base case.

For the inductive step, assume that the claim is true for $j$, where $1 \leq j < \lceil \mu|V_1| \rceil$, we shall show it is true for $j+1$. Let us assume $j$ distinct elements from $V_1$ have already been enumerated, and a total of at least $(j+1)\alpha$ elements have been sampled (of which at least $j$ are from $V_1$). The probability an element from $V_1$ that was not already enumerated is sampled is 
$ (|V_1| - j)/n.$
Therefore, the probability that exactly $j$ unique elements from $V_1$ have been sampled is at most
\[\Big(1-\frac{|V_1|-j}{n}\Big)^{(j+1)\alpha -j} <(1-\mu(1-\mu))^{(j+1)\alpha -j},\]
as $|V_1|-j > |V_1|-\mu|V_1| \geq \mu n(1-\mu)$. Then
\[(1-\mu(1-\mu))^{(j+1)\alpha -j} \leq \frac{q^{j+1}}{(1-\mu(1-\mu))^j}\leq \frac{q^{j+1}}{(q^{\frac{1}{2}})^j} = q^{\frac{j+2}{2}},\]
as $\alpha = \lceil \log_{1-\mu(1-\mu)}q \rceil \geq \log_{1-\mu(1-\mu)}q $ and as $q\leq (1-\mu(1-\mu))^2$. Therefore, the probability that there are at least $j+1$ elements from $V_1$ in these sampled tuples is at least $1-q^{(j+2)/2}$. Then by the inductive hypothesis, the probability that at least $j+1$ elements from $V_1$ are enumerated is at least \[\Big(\prod_{i=1}^j(1-q^{\frac{j+1}{2}})\Big)\cdot(1-q^{\frac{j+2}{2}}) = \prod_{i=1}^{j+1}(1-q^{\frac{i+1}{2}}) \] as required. 

Finally, by Claim \ref{probability bound}, the probability that at least $\lceil \mu|V_1| \rceil$  many distinct elements from $V_1$ are enumerated is at least \[\prod_{i=1}^{\lceil \mu|\phi(\mathcal{D})| \rceil }(1-q^{\frac{i+1}{2}}) \geq 1-3q^{\frac{1}{2}}.\] \end{claimproof}

By Claim \ref{printing probability} the probability that $S = V_1$ if $|V_1| \geq \mu n$ is at least
$(1-3q^{\frac{1}{2}}) \geq \delta$
by the choice of $q$. This completes the proof.
\end{proof}

We now use Lemma \ref{lemma: partitioned set enumeration} to prove the following theorem.

\begin{thm}\label{thm: local queries enumeration}
Let $\phi(\bar{x}) \in  \operatorname{FO}[\sigma]$ be a local query with $k$ free variables and let $\gamma \in (0,1)$. There exists an algorithm that is given a $\sigma$-db $\mathcal{D}$ as an input, that after a constant time preprocessing phase, enumerates a set $S$ (with no duplicates) with constant delay between any two consecutive outputs, such that:
\begin{enumerate}
\item $S \subseteq \phi(\mathcal{D})$, and
\item if $|\phi(\mathcal{D})|| \geq \gamma |D|^k$ (i.e. the number of answers to the query is larger than a fixed fraction of the total possible number of answers), then with probability at least $2/3$, $S = \phi(\mathcal{D})$.
\end{enumerate}
\end{thm}
\begin{proof}
Given a tuple $\bar{a} \in |D|^k$ we can test in constant time whether $\bar{a} \in \phi(\mathcal{D})$ or $\bar{a} \not\in \phi(\mathcal{D})$ by Lemma \ref{lemma: local membership testing}. We can partition the set $D^k$ into two sets based on whether a tuple is in $\phi(\mathcal{D})$ or not.
Therefore the algorithm from Lemma \ref{lemma: partitioned set enumeration} (with $\delta = 2/3$, $\mu=\gamma$, $V=|D|^k$, $V_1= \phi(\mathcal{D})$ and $V_2=|D|^k \setminus \phi(\mathcal{D})$) meets the requirements in the theorem statement.
\end{proof}

In our algorithms, in order to achieve constant preprocessing time and constant delay we require the number of answers to the query to be some fixed fraction of the total possible number of answers. Otherwise, with high probability the algorithm would not sample an answer in the enumeration phase and the algorithm would stop.

It seems natural to expect that for queries occurring in practice, the elements of an answer tuple are 
%close together 
within a small distance of each other in the input database (i.\,e.\ the $r$-neighbourhood of the answer has few connected components). In such scenarios, we can strengthen our main theorem by reducing the number of answers required to output all answers to the query with high probability.

\begin{thm}\label{theorem: strengthened local main}
Let $\phi(\bar{x}) \in  \operatorname{FO}[\sigma]$ be a local query with locality radius $r$ and let $\gamma \in (0,1)$. 
Let $c : =\operatorname{conn}(\phi, d)$, i.e the maximum number of connected components in the $r$-neighbourhood of a tuple $\bar{a} \in \phi(\mathcal{D})$ for any $\sigma$-db $\mathcal{D}$.
%Let $\phi'$ be the formula in Hanf normal form that is $d$-equivalent to $\phi$. Let $c$ be the maximum number of connected components in the $r$-types that appear in $\phi'$. 
There exists an algorithm that, given a $\sigma$-db $\mathcal{D}$ as input, after a constant time preprocessing phase enumerates a set $S$ (with no duplicates) with constant delay between any two consecutive outputs, such that the following hold.
\begin{enumerate}
\item $S \subseteq \phi(\mathcal{D})$, and
\item if $|\phi(\mathcal{D})|| \geq \gamma |D|^c$, then with probability at least $2/3$, $S = \phi(\mathcal{D})$.
\end{enumerate}
\end{thm}

We defer the proof of Theorem \ref{theorem: strengthened local main} to Section \ref{sec: proof of theorem strengthened local main}.

%lower bounds - to have constant preprocessing we need a linear number of answers (if we use random sampling)
%
%to have sublinear preprocessing we need $|q(\mathcal{D})|/n^k$ to be $O(1)/o(n)$ I think? 

\section{Enumerating Answers to General First-Order Queries}\label{sec: non local queries}

We now shift our focus to enumerating answers to general FO queries, now they can be non-local in the sense that we can not check if a tuple is an answer to the query by only looking at its neighbourhood. We are aiming at sublinear preprocessing time hence we cannot read the whole input database and therefore will need to sacrifice some accuracy. We allow our algorithms to enumerate `close' answers as well as actual answers. We start this section by defining our notion of approximation before proving our main result.

\subsection{Our Notion of Approximation}

We shall start by defining our notion of closeness.

\begin{defi}[$\epsilon$-close answers to FO queries]\label{def: closeness}
Let $\mathcal{D} \in \mathbf{C}$ be a $\sigma$-db and let $\epsilon \in (0,1]$. Let $\phi(\bar{x}) \in  \operatorname{FO}[\sigma]$ be a query with $k$ free variables and Hanf locality radius $r$.  A tuple $\bar{a} \in D^k$ is \emph{$\epsilon$-close to being an answer of $\phi$ on $\mathcal{D}$ and $\mathbf{C}$} if $\mathcal{D}$ can be modified (with tuple insertions and deletions) into a $\sigma$-db $\mathcal{D}' \in \mathbf{C}$ with at most $\epsilon d |D|$ modifications (i.e $\operatorname{dist}(\mathcal{D}, \mathcal{D}') \leq \epsilon d |D|$) such that $\bar{a} \in \phi(\mathcal{D}')$ and the $r$-type of $\bar{a}$ in $\mathcal{D}'$ is the same as the $r$-type of $\bar{a}$ in $\mathcal{D}$.
%there exists a $\sigma$-db $\mathcal{D}_1 \in \mathbf{C}$ with $|D|$ elements such that 
%\begin{enumerate}
%\item $\operatorname{dist}(\mathcal{D}, \mathcal{D}_1) \leq \epsilon d |D|$, and 
%\item there exists a tuple $ \bar{b} \in \phi(\mathcal{D}_1)$ with the same $r$-type as $\bar{a}$.
%\end{enumerate} 

We denote the set of all tuples that are $\epsilon$-close to being an answer of $\phi$ on $\mathcal{D}$ and $\mathbf{C}$ as $\phi(\mathcal{D}, \mathbf{C}, \epsilon)$. Note that $\phi(\mathcal{D}) \subseteq \phi(\mathcal{D}, \mathbf{C}, \epsilon)$.
\end{defi}

We shall illustrate Definition \ref{def: closeness} in the following example.

\begin{exa}\label{running example}
%Let $\mathbf{C}$ be the class of simple graphs with bounded degree $d \in \mathbb{N}$. 
	On the class $\mathbf G_d$, consider the isomorphism types
	 $\tau_1$, $\tau_2$ and $\tau_3$ of the $2$-neighbourhoods $(N_1, (c_1,c_2))$, $(N_2, (c_1,c_2))$ and $(N_3, (c_1,c_2))$ shown in Figure~\ref{fig:taus}. Let $\phi\in\operatorname{FO}[\{E\}]$ be given by $\phi(x,y) :=  \operatorname{sph}_{\tau_1}(x,y) \lor (\operatorname{sph}_{\tau_2}(x,y) \land \lnot(\exists z \exists w\operatorname{sph}_{\tau_1}(z,w))).$ This formula might be useful in scenarios where ideally we want to return pairs of vertices with a specific $2$-type $\tau_1$ but if there is no such pair then returning vertex pairs with a similar $2$-type will suffice.
	 
Let $\mathcal{G} \in \mathbf{G}_d$ be a graph on $n$ vertices and $\epsilon \in (0,1]$. First observe that for any pair $(u,v) \in V(\mathcal{G})^2$ with 2-type $\tau_1$, $(u,v) \in \phi(\mathcal{G})$ and hence $(u,v) \in \phi(\mathcal{G},\mathbf{G}_d, \epsilon)$. 

Assume $(u,v) \in V(\mathcal{G})^2$ has 2-type $\tau_2$. Then $(u,v) \in \phi(\mathcal{G})$ if and only if $\mathcal{G}$ contains no vertex pair of 2-type $\tau_1$. 
The pair $(u,v)$ is in $\phi(\mathcal{G},\mathbf{G}_d, \epsilon)$ if and only if $\mathcal{G}$ can be modified (with edge modifications) into a graph $\mathcal{G'} \in \mathbf{G}_d$ with at most 
$\epsilon d n$ modifications such that $(u,v) \in \phi(\mathcal{G'})$ and the 2-type of $(u,v)$ in $\mathcal{G'}$ is still $\tau_2$.

For example if $\mathcal{G}$ is at distance at most $\epsilon d n - 4d - 6$ (assuming that $n$ is large enough such that $\epsilon d n - 4d - 6 >0$) from a graph $\mathcal{G''} \in \mathbf G_d$ such that $\mathcal{G''} \models \exists x \exists y \operatorname{sph}_{\tau_2}(x,y) \land \lnot(\exists z \exists w \operatorname{sph}_{\tau_1}(z,w)) $ then $(u,v) \in\phi(\mathcal{G},\mathbf{G}_d, \epsilon)$. To see this let us assume that such a graph $\mathcal{G''} $ exists. Note that as $\mathbf{G}_d$ is closed under isomorphism we can assume that $\mathcal{G} $ and $\mathcal{G''} $ are on the same vertices. Then if $(u,v)$ has 2-type $\tau_2$ in $\mathcal{G''} $, $(u,v) \in\phi(\mathcal{G},\mathbf{G}_d, \epsilon)$ since $\epsilon d n - 4d - 6 \leq \epsilon d n$. So let us assume that $(u,v)$ does not have 2-type $\tau_2$ in $\mathcal{G''} $. Let $(u_1, v_1)  \in V(\mathcal{G''})^2 $ have 2-type $\tau_2$ (we know one exists). Then we remove every edge that has $u$, $v$, $u_1$ or $v_1$ as an endpoint (there are at most $2d +3$ such edges), and then for each edge we removed we insert the same edge back in but swapping any endpoint $u$ to $u_1$ and $v$ to $v_1$ and vice versa (this requires at most $2(2d+3)$ many edge modifications in total). By doing this we have essentially just swapped the labels of the vertices $u$ and $u_1$ and $v$ and $v_1$. Hence in the resulting graph $\mathcal{G'}$, $(u,v)$ has 2-type $\tau_2$ and $\mathcal{G'}$ still contains no pair of vertices with 2-type $\tau_1$. Therefore $(u,v) \in \phi(\mathcal{G'})$, and the distance between $\mathcal{G}$ and $\mathcal{G'}$ is at most $\epsilon d n - 4d - 6 + 2(2d+3) = \epsilon d n$.

%, $(u,v)$ does not have 2-type $\tau_2$. Let $(u_1, v_1)  \in \mathcal{G''}^2 $ have 2-type $\tau_2$ (we know one exists). Then if we remove every edge that has $u$, $v$, $u_1$ or $v_1$ as an endpoint (there are at most $2n +3$ such edges), and then for each edge we removed we insert the same edge back in but swapping any endpoint $u$ to $u_1$ and $v$ to $v_1$ and vice versa. Then if $\mathcal{G'}$ is the resulting graph, $(u,v)$ has 2-type $\tau_2$ in $\mathcal{G'}$, $(u,v) \in \phi(\mathcal{G'})$, and the distance between $\mathcal{G}$ and $\mathcal{G'}$ is at most $\epsilon d n$.\polly!

% with at most $\epsilon d n - 2d$ modifications such that $\mathcal{G'}$ satisfies $\exists x \exists y \operatorname{sph}_{\tau_2}(x,y) \land \lnot(\exists z \exists w \operatorname{sph}_{\tau_1}(z,w)) $. \polly{explain more}

Finally, for any pair $(u,v) \in V(\mathcal{G})^2$ with 2-type $\tau_3$, $(u,v) \not\in \phi(\mathcal{G})$ and $(u,v) \not\in \phi(\mathcal{G},\mathbf{G}_d, \epsilon)$ as for every $\mathcal{G'} \in \mathbf{G}_d$ there does not exist a pair with 2-type $\tau_3$ that is in $\phi(\mathcal{G'})$.

\end{exa}

The set $\phi(\mathcal{D},\mathbf{C}, \epsilon)$ contains all tuples that are
$\epsilon$-close to being answers to $\phi$. A tuple $\bar{a} \in D^k$\ is in
$\phi(\mathcal{D},\mathbf{C}, \epsilon)$ if only a relatively small (at most $\epsilon d n$) number of modifications to $\mathcal{D}$ are needed to make $\bar{a}$ an answer to $\phi$ without changing $\bar{a}$'s neighbourhood type. This can be seen as a notion of \emph{structural} approximation. One might be tempted to define $\phi(\mathcal{D},\mathbf{C}, \epsilon)$ differently, namely as the set of tuples that can be turned into an answer to $\phi$ on $\mathcal{D}$ (without necessarily preserving the neighbourhood type) with at most $\epsilon d n$ modifications to $\mathcal D$.
However, if $\phi(\mathcal{D})\neq\emptyset$, say, $\bar{a}\in \phi(\mathcal{D})$, 
	then we can turn any tuple $\bar{b} \in D^k$ into an answer 
	for $\phi$ on $\mathcal{D}$ with only a constant number of modifications.
	This can be done by exchanging $\bar{b}$'s $r$-neighbourhood with $\bar{a}$'s, for some $r$ depending on $\phi$. This is not meaningful.
%	But then $\phi(\mathcal{D},\mathbf{C}, \epsilon)$ would be equal to $D^k$ 
%	(if  $D$ is sufficiently large), 
%	which is not meaningful.

Let $\chi$ be as in (\ref{normal-form}) of Lemma \ref{normal-form lemma} for $\phi$. Note that only tuples with a neighbourhood type that appears in $\chi$ can be in the set $ \phi(\mathcal{D}, \mathbf{C}, \epsilon)$.
Nevertheless, the difference $|\phi(\mathcal{D},\mathbf{C},
\epsilon)|-|\phi(\mathcal{D})|$ can be unbounded. 
The following example demonstrates this.

\begin{exa}\label{ex:many-eps-close} 
	Let $\phi$, $\tau_1$ and $\tau_2$ be as in Example \ref{running example}. For $m \in \mathbb{N}_{\geq 1}$, let $\mathcal{G}_{1,m}$ be the graph that contains $m$ disjoint copies of $\tau_2$ and 1 disjoint copy of $\tau_1$. Note that $\mathcal{G}_{1,m}$ has $n=8(m+1)$ vertices. The graph $\mathcal{G}_{1,m}$ can be modified with one edge modification to form a graph which satisfies $\exists x \exists y\operatorname{sph}_{\tau_2}(x,y) \land \lnot(\exists z \exists w\operatorname{sph}_{\tau_1}(z,w))$ without modifying the 2-type of any pair $(u,v) \in V(\mathcal{G}_{1,m})^2$ with 2-type $\tau_2$ in $\mathcal{G}_{1,m}$.
 Therefore if $1 \leq \epsilon d n$ then every pair $(u,v) \in V(\mathcal{G}_{1,m})^2$ with 2-type $\tau_2$ is in $\phi(\mathcal{G}_{1,m},\mathbf{G}_d,\epsilon)$.
% $\mathcal{G}_{1,m}$ is $\epsilon$-close to satisfying $\exists x \exists y\operatorname{sph}_{\tau_2}(x,y) \land \lnot(\exists z \exists w\operatorname{sph}_{\tau_1}(z,w))$ and $\tau_2$ is $\epsilon$-close to being relevant for $\mathcal{G}_{1,m}$ and $\phi$ on $\mathbf{G}_d$. 
 Hence, assuming $1 \leq \epsilon d n$ we have $|\phi(\mathcal{G}_{1,m},\mathbf{G}_d,\epsilon)|-|\phi(\mathcal{G}_{1,m})|=m+1-1=\Theta(n)$.
\end{exa}

While $\phi(\mathcal{D},\mathbf{C}, \epsilon)$ is a \emph{structural} approximation
of $\phi(\mathcal D)$, Example~\ref{ex:many-eps-close} illustrates that 
it may not be a \emph{numerical} approximation. However, in scenarios where
the focus lies on structural closeness, this might not be an issue.

We say that the problem $\operatorname{Enum}_{\mathbf{C}}(\phi)$ can be \emph{solved approximately} with $\mathcal{O}(H(n))$ preprocessing time and constant delay for answer threshold function $f(n)$, if for every parameter $\epsilon \in (0,1]$, there exists an algorithm, which is given oracle access to an input database $\mathcal{D} \in \mathbf{C}$ and $|D|=n$ as an input, that proceeds in two steps.
\begin{enumerate}
\item A preprocessing phase that runs in time $\mathcal{O}(H(n))$, and
\item an enumeration phase that enumerates a set $S$ of distinct tuples with constant delay between 
	any two consecutive outputs.
\end{enumerate}
 Moreover, we require that with probability at least $2/3$, $S \subseteq \phi(\mathcal{D}) \cup\phi(\mathcal{D},\mathbf{C}, \epsilon)$ and, if $|\phi(\mathcal{D})| \geq f(n)$, then $\phi(\mathcal{D}) \subseteq S$.
The algorithm can make oracle queries of the form $(R,i,j)$ as discussed in Section \ref{Prelim PT} which allows us to explore bounded radius neighbourhoods in constant time. 
We call such an algorithm an \emph{$\epsilon$-approximate enumeration algorithm}.

\subsection{Main Results}

Before proving our main result of this section on the approximate enumeration of general first-order queries, we start by proving the following lemma. In this lemma, we show that for a given database $\mathcal{D}$ and FO query $\phi$ we can compute a set of neighbourhood types in polylogarithmic time, that with high probability only contains the neighbourhood types of tuples that are answers or close to being answers to $\phi$ on $\mathcal{D}$. To compute this set we write $\phi$ in the form (\ref{normal-form}) as in Lemma~\ref{normal-form lemma} and then run property testers on the sentence parts to determine with high probability whether tuples with the corresponding $r$-type (the $r$-type that appears in the sphere-formula) are answers to $\phi$ on the input database or are far from being an answer to $\phi$ on the input database.

\begin{lem}\label{lemma: computing rel r-types}
Let $\phi(\bar{x}) \in \operatorname{FO}[\sigma]$ with $|\bar{x}|=k$ and Hanf locality radius $r$ and let $\epsilon \in (0,1]$. There exists an algorithm $\mathbb{A}_{\epsilon}$, which, given oracle access to a $\sigma$-db $\mathcal{D} \in \mathbf{C}_d^t$ as input along with $|D|=n$, computes a set $T$ of $r$-types with $k$ centres such that with probability at least $5/6$, for any $\bar{a} \in D^k$,
\begin{enumerate}
\item if $\bar{a} \in \phi(\mathcal{D})$, then the $r$-type of $\bar{a}$ in $\mathcal{D}$ is in $T$, and
\item if $\bar{a} \in D^k \setminus \phi(\mathcal{D},\mathbf{C}_d^t, \epsilon )$, then the $r$-type of $\bar{a}$ in $\mathcal{D}$ is not in $T$.
\end{enumerate} 
Furthermore, $\mathbb{A}_{\epsilon}$ runs in polylogarithmic time.
\end{lem}
\begin{proof}
If $n < 8k/ \epsilon$ then we do a full check of $\mathcal{D}$ and form the set $T$ exactly. Otherwise, $\mathbb{A}_{\epsilon}$ starts by computing the formula $\chi(\bar{x})$ that is $d$-equivalent to $\phi$ and is in the form (\ref{normal-form}) as in Lemma~\ref{normal-form lemma}. Let $m$ be as in~Lemma~\ref{normal-form lemma}.  By Theorem \ref{FO testability}, any sentence definable in FO is uniformly testable on $\mathbf{C}_d^t$ in polylogarithmic time. Hence for every $i \in [m]$ there exists an $\epsilon/2$-tester that runs in polylogarithmic time and with probability at least ${2}/{3}$ accepts if the input satisfies $\exists \bar{x} \operatorname{sph}_{\tau_i}(\bar{x}) \land \psi^s_i$ and rejects if the input is $\epsilon/2$-far from satisfying $\exists \bar{x} \operatorname{sph}_{\tau_i}(\bar{x}) \land \psi^s_i$. We can amplify this probability to $({5}/{6})^{{1}/{m}}$ by repeating the tester a constant number of times and we denote the resulting $\epsilon/2$-tester as~$\pi_i$. Next, $\mathbb{A}_{\epsilon}$ computes the set $T$ as follows.
\begin{enumerate}
\item Let $T= \emptyset$.
\item For each $i \in [m]$, run $\pi_i$ with $\mathcal{D}$ as input, and if $\pi_i$ accepts, then add $\tau_i$ to $T$.
\end{enumerate}

By Lemma~\ref{normal-form lemma}, $\chi(\bar{x})$ can be computed in constant time (only dependent on $d$, $\|\phi\|$ and $\|\sigma\|$). Moreover, each $\epsilon/2$-tester $\pi_i$ runs in polylogarithmic time. Since $m$ is a constant, $\mathbb{A}_{\epsilon}$ runs in polylogarithmic time.  

It now only remains to prove correctness. Let $\bar{a} \in D^k$ and let $\tau$ be the $r$-type of $\bar{a}$ in $\mathcal{D}$. Let us assume that each $\pi_i$ correctly accepts if $\mathcal{D}$ satisfies $\exists \bar{x} \operatorname{sph}_{\tau_i}(\bar{x}) \land \psi^s_i$ and correctly rejects if $\mathcal{D}$ is $\epsilon/2$-far from satisfying $\exists \bar{x} \operatorname{sph}_{\tau_i}(\bar{x}) \land \psi^s_i$, which happens with probability at least $(5/6)^{(1/m) \cdot m}={5}/{6}$.

First let us assume that $\bar{a} \in \phi(\mathcal{D})$. We shall show that $\tau \in T$. Since $\mathcal{D} \models \phi(\bar{a})$, there exists at least one $i \in [m]$ such that $\mathcal{D} \models  \operatorname{sph}_{\tau_i}(\bar{a}) \land \psi^s_i$ (as $\phi$ is $d$-equivalent to $\chi(\bar{x}) =\bigvee_{i\in [m]}\Big(\operatorname{sph}_{\tau_i}(\bar{x}) \land \psi^s_i \Big)$). Hence, $\mathcal{D} \models  \exists \bar{x}\operatorname{sph}_{\tau_i}(\bar{x}) \land \psi^s_i$ and as we are assuming $\pi_i$ correctly accepted, the $r$-type $\tau_i$ will have been added to $T$. Since $\mathcal{D} \models  \operatorname{sph}_{\tau_i}(\bar{a})$, $\tau_i = \tau$, and therefore $\tau \in T$.

Now let us assume that $\bar{a} \in D^k \setminus \phi(\mathcal{D},\mathbf{C}_d^t, \epsilon )$. We shall show that $\tau \not\in T$. For a contradiction let us assume that $\tau \in T$ and hence there must exist some $i \in [m]$ such that $\mathcal{D}$ is $\epsilon/2$-close to satisfying $\exists \bar{x}\operatorname{sph}_{\tau_i}(\bar{x}) \land \psi^s_i$ on $\mathbf{C}_d^t$ and $\tau_i = \tau$. By definition there exists a $\sigma$-db $\mathcal{D}' \in  \mathbf{C}_d^t$ such that $\mathcal{D}' \models \exists \bar{x}\operatorname{sph}_{\tau_i}(\bar{x}) \land \psi^s_i$ and $\operatorname{dist}(\mathcal{D}, \mathcal{D}') \leq \epsilon d n /2$. Since any property defined by a FO sentence on $\mathbf{C}_d^t$ is closed under isomorphism we can assume that $\mathcal{D}'$ can be obtained from $\mathcal{D}$ with at most $\epsilon d n /2$ tuple modifications. If in $\mathcal{D}'$ the $r$-type of $\bar{a}$ is no longer $\tau$ then we can modify $\mathcal{D}'$ with at most $4dk$ tuple modifications into a $\sigma$-db $\mathcal{D}'' \in  \mathbf{C}_d^t$ such that the $r$-type of $\bar{a}$ is $\tau$ in $\mathcal{D}'' $ and $\mathcal{D}'' \cong \mathcal{D}'$ (and hence $\bar{a} \in \phi(\mathcal{D}'')$). To do this we choose a tuple $\bar{b}$ whose $r$-type is $\tau$ in $\mathcal{D}'$ and for any tuple that contains an element from $\bar{a}$ or $\bar{b}$, delete it and add back the same tuple but with the elements from $\bar{a}$ exchanged for the corresponding elements from $\bar{b}$ and vice versa. This requires at most $4dk$ tuple modifications. Hence $\operatorname{dist}(\mathcal{D}, \mathcal{D}'') \leq\epsilon d n/2 +4dk \leq \epsilon d n$ if $n \geq 8k / \epsilon$ (which we can assume as otherwise we do a full check of $\mathcal{D}$ and compute $T$ exactly) and so by definition $\bar{a} \in \phi(\mathcal{D},\mathbf{C}_d^t, \epsilon )$ which is a contradiction. Therefore $\tau \not\in T$.

Hence with probability at least $5/6$, for every $\bar{a} \in D^k$, if $\bar{a} \in \phi(\mathcal{D})$, then the $r$-type of $\bar{a}$ in $\mathcal{D}$ is in $T$, and  if $\bar{a} \in D^k \setminus \phi(\mathcal{D},\mathbf{C}_d^t, \epsilon )$, then the $r$-type of $\bar{a}$ in $\mathcal{D}$ is not in $T$.
\end{proof}

We now use Lemmas \ref{lemma: partitioned set enumeration} and \ref{lemma: computing rel r-types} to prove our main result of this section (Theorem \ref{thm:non-local enumeration main}).

\begin{thm}\label{thm:non-local enumeration main}
Let $\phi(\bar{x}) \in \operatorname{FO}[\sigma]$ where $|\bar{x}|=k$. Then $\operatorname{Enum}_{\mathbf{C}_d^t}(\phi)$ can be solved approximately with polylogarithmic preprocessing time and constant delay for answer threshold function $f(n)=\gamma n^k$ for any parameter $\gamma \in (0,1)$.
\end{thm}
\begin{proof}
Let $\mathcal{D} \in \mathbf{C}_d^t$ with $|D|=n$, let $\epsilon \in (0,1]$ and let $\gamma \in (0,1)$. We shall construct an $\epsilon$-approximate enumeration algorithm for $\operatorname{Enum}_{\mathbf{C}_d^t}(\phi)$ that has answer threshold function $f(n)=\gamma n^k$, polylogarithmic preprocessing time and constant delay.

In the preprocessing phase, the algorithm starts by running the algorithm from Lemma \ref{lemma: computing rel r-types} on $\mathcal{D}$ to compute a set $T$ of $r$-types with $k$ centres.
Then the algorithm from the proof of Lemma \ref{lemma: partitioned set enumeration} with $\mu=\gamma$, $\delta =  5/6$, $V =D^k$, $V_1=\{\bar{a} \in D^k \mid \text{the r-type of } \bar{a} \text{ in } \mathcal{D}\text{ is in } T \}$ and $V_2 = D^k \setminus V_1$ is run. 

By Lemma~\ref{lemma: computing rel r-types}, the set $T$ is computed in polylogarithmic time. Hence as the preprocessing phase from the proof of Lemma \ref{lemma: partitioned set enumeration} runs in constant time, the whole preprocessing phase runs in polylogarithmic time. By Lemma \ref{lemma: partitioned set enumeration} there is constant delay between any two consecutive outputs. This concludes the analysis of the running time. We now prove correctness.

Let $S$ be the set of tuples enumerated. By Lemma \ref{lemma: partitioned set enumeration} no duplicates are enumerated and $S \subseteq V_1=\{\bar{a} \in D^k \mid \text{the r-type of } \bar{a} \text{ is in } T \}$. By Lemma \ref{lemma: computing rel r-types}, with probability at least $5/6$, for every $\bar{a} \in D^k$, if $\bar{a} \in \phi(\mathcal{D})$, then the $r$-type of $\bar{a}$ in $\mathcal{D}$ is in $T$, and  if $\bar{a} \in D^k \setminus \phi(\mathcal{D},\mathbf{C}_d^t, \epsilon )$, then the $r$-type of $\bar{a}$ in $\mathcal{D}$ is not in $T$. Therefore with probability at least $5/6$, $\phi(\mathcal{D}) \subseteq V_1$ and $V_1 \subseteq \phi(\mathcal{D},\mathbf{C}, \epsilon)$. Hence with probability at least $5/6 > 2/3$, $S \subseteq  \phi(\mathcal{D}) \cup\phi(\mathcal{D},\mathbf{C}, \epsilon)$ as required. 
%Moreover, by Lemma \ref{lemma: partitioned set enumeration} if $|V_1| \geq \gamma |V|$ then, with probability at least $5/6$, $S=V_1$.
 As previously discussed with probability at least $5/6$, $\phi(\mathcal{D}) \subseteq V_1$. Note that if $\phi(\mathcal{D}) \subseteq V_1$, then $|V_1| \geq |\phi(\mathcal{D})|$. If we assume that $\phi(\mathcal{D}) \subseteq V_1$ and $|\phi(\mathcal{D})| \geq \gamma n^k = \gamma |V|$, then $|V_1| \geq \gamma |V|$ and by Lemma \ref{lemma: partitioned set enumeration} with probability at least $5/6$, $S=V_1$ and hence $\phi(\mathcal{D}) \subseteq S$. Therefore the probability that $\phi(\mathcal{D}) \subseteq S$ if $|\phi(\mathcal{D})| \geq \gamma n^k$ is at least $(5/6)^2 > 2/3$ as required. This completes the proof.
%  By combining the probability that the set $T$ was correctly constructed and the probability that $\phi(\mathcal{D}) \subseteq S$  if $|\phi(\mathcal{D})| \geq \gamma n^k$ and $\phi(\mathcal{D}) \subseteq V_1$, the probability that $\phi(\mathcal{D}) \subseteq S$ if $|\phi(\mathcal{D})| \geq \gamma n^k$ is at least $(5/6)^2 > 2/3$ as required. This completes the proof.
\end{proof}

%We sketch the proof of Theorem \ref{theorem: strengthened main}. First note that to discover a $k$-tuple $\bar{a} \in \phi(\mathcal{D}, \mathbf{C}, \epsilon)$ we only actually need to sample an element from each connected component in the neighbourhood of $\bar{a}$. Using this fact and the definition of $c_m$ we can sample $c_m$-tuples instead of $k$-tuples as is done in the proof of Theorem \ref{thm:non-local enumeration main}.  We can show that each $c_m$-tuple gives at most a constant number of $k$-tuples from $\phi(\mathcal{D})$ and therefore if $|\phi(\mathcal{D})| \geq \gamma n^{c_m}$, then the number of $c_m$-tuples that lead to at least one tuple in $\phi(\mathcal{D})$ is a constant fraction of $\gamma n^{c_m}$. Hence there is a sufficient number of `good' $c_m$-tuples to allow us to sample one in the preprocessing phase with high probability. We can then start enumerating and sample more tuples between outputs as in the proof of Theorem \ref{thm:non-local enumeration main}.
%Before stating our next result we return to our running example to demonstrate the reduction in the answer threshold function given in Theorem \ref{theorem: strengthened main}.
As discussed in Section \ref{sec: local queries}, it is natural for us to expect that for queries that occur in practice, the neighbourhood of the answer tuple has few connected components. We saw that for local FO queries, in such scenarios we can reduce the number of answers required to output all answers to the query with high probability (Theorem \ref{theorem: strengthened local main}). The following theorem shows how we can reduce the answer threshold function for general FO queries.

\begin{thm}\label{theorem: strengthened non local main}
Let $\phi(\bar{x}) \in \operatorname{FO}[\sigma]$ and let $c:=\operatorname{conn}(\phi,d)$. Then the problem $\operatorname{Enum}_{\mathbf{C}^t_d}(\phi)$ can be solved approximately with polylogarithmic preprocessing time and constant delay for answer threshold function $f(n)=\gamma n^{c}$ for any parameter $\gamma \in (0,1)$.

\end{thm}

We defer the proof of Theorem \ref{theorem: strengthened non local main} to Section \ref{sec: proof of theorem strengthened local main}.

\section{Proofs of Theorems \ref{theorem: strengthened local main} and \ref{theorem: strengthened non local main}}\label{sec: proof of theorem strengthened local main}
%Put a note on why we need to introduce this notation.

Before we prove Theorems \ref{theorem: strengthened local main} and \ref{theorem: strengthened non local main} we start with some definitions (which are based on those introduced by Kazana and Segoufin in \cite{KazanaS11}) and some lemmas. 

For each type $\tau \in T_r^{\sigma,d}(k)$ we fix a representative for the corresponding $r$-type and fix a linear order among its elements (where, for technical reasons, the centre elements always come first). This way, we can speak of the first, second, $\dots$, element of an $r$-type. Let $\mathcal{D}$ be a $\sigma$-db and let $\bar{a}$ be a tuple in $\mathcal{D}$ with $r$-type $\tau$. For technical reasons, if there are multiple isomorphism mappings from the $r$-neighbourhood of $\bar{a}$ to the fixed representative of $\tau$, we use the isomorphism mapping which is of smallest lexicographical order (recall that we assume that $\mathcal{D}$ comes with a linear ordering on its elements).
The cardinality of $\tau$, denoted as $|\tau|$, is the number of elements in its representative.

Let $\mathcal{D}$ be a $\sigma$-db  and $\bar{a}$ be a tuple of elements from $\mathcal{D}$. We say that $\bar{a}$ is \emph{$r$-connected} if the $r$-neighbourhood of $\bar{a}$ in $\mathcal{D}$ is connected.

Let $s\in \mathbb{N}$, let $F=(\alpha_2,\dots,\alpha_m)$ be a sequence of elements from $[d^{s+1}]$ (recall that the maximum size of an $s$-neighbourhood is $d^{s+1}$), and let $\bar{x}=(x_1, \dots, x_m)$ be a tuple. We write $\bar{x} = F(x_1)$ for the fact that, for $j \in \{2,\dots, m\}$, $x_j$ is the $\alpha_j$-th element of the $s$-neighbourhood of $x_1$. We call each such $F$ an \emph{$s$-binding of $\bar{x}$}. Given $s$-type $\tau$, we say that an $s$-binding $F$ of $\bar{x}$ is \emph{$r$-good for $\tau$} if $F(x_1)$ is $r$-connected for every $x_1$ with type $\tau$. 
%We let $\mathcal{F}_s^m$ denote the set of all $s$-bindings of tuples with $m$ elements.

For a given tuple $\bar{x}= (x_1,\dots,x_k)$, an $r$-split of $\bar{x}$ is a set of triples \\ $C=\{(C_1, F_1, \tau_1), \dots, (C_{\ell}, F_{\ell}, \tau_{\ell})\}$ where for each $i \in [\ell]$
\begin{itemize}
\item $\emptyset \neq C_i \subseteq \bar{x}$,
$C_i \cap C_j = \emptyset$ for $i \neq j\in [\ell]$ and
$\bigcup_{1 \leq i \leq \ell}C_i = \{x_1,\dots,x_k\}$,
%\item $\bar{x}_a = (x_{a_1},\dots, x_{a_c})$ and $\bar{x}_b = (x_{b_1},\dots, x_{b_{k-c}})$ where $a_1, \dots,a_c,b_1,\dots,b_{k-c} \in [k]$, \\$\{a_1, \dots,a_c\} \cap \{b_1,\dots,b_{k-c} \} =\emptyset$ and $\{a_1, \dots,a_c\} \cup \{b_1,\dots,b_{k-c} \} =[k]$,
%\item $\{x_1^1,\dots, x_c^1\}\cap \{x_1^2,\dots, x_{k-c}^2\} = \emptyset$ and $\{x_1^1,\dots, x_c^1\}\cup \{x_1^2,\dots, x_{k-c}^2\} = \{x_1,\dots,x_k\}$, 
\item $\tau_i$ is a $3rk$-type with $1$ centre, and
\item $F_i=(\alpha_2,\dots,\alpha_{|C_i|})$ is a $3rk$-binding of a tuple with $|C_i|$ elements such that for each $j \in \{2,\dots,|C_i|\}$, $\alpha_j \in [|\tau_i|]$ and $F_i$ is $r$-good for $\tau_i$.

\end{itemize}
We write $\bar{x}^i$ to represent the variables from $C_i$, $x_1^i$ to represent the most significant variable from $C_i$ (i.e the variable in $C_i$ which appears first in the tuple $\bar{x}$), $x_2^i$ to represent the second most significant variable from $C_i$ (i.e the variable in $C_i$ which appears second in the tuple $\bar{x}$) and so on. We define the formula \[\operatorname{Split}_{r}^C(\bar{x}) := \bigwedge_{1\leq i \neq j \leq \ell}(N_r(\bar{x}^i) \cap N_r(\bar{x}^j) = \emptyset) \wedge \bigwedge_{(C_i,F_i,\tau_i) \in C}(\bar{x}^i=F_i(x_1^i) \land  \operatorname{sph}_{\tau_i}(x_1^i)).\]
 We let $S_r^{\sigma,d}(k)$ denote the set of $r$-splits of tuples with $k$ elements for $\sigma$-dbs with degree at most $d$. We denote the cardinality of $S_r^{\sigma,d}(k)$ as $s(r,k)$.

\begin{rem}\label{remark: unique split} For any $r,k \in \mathbb{N}$, $\sigma$-db $\mathcal{D}$ and tuple $\bar{a} \in D^k$ there exists exactly one $r$-split $C$ such that $\mathcal{D} \models \operatorname{Split}_{r}^C(\bar{a})$.
\end{rem}

Let $\mathcal{D}$ be a $\sigma$-db, let $r,k,c \in \mathbb{N}$ where $c \leq k$ and let $C$ be an $r$-split for a tuple with $k$ elements. For tuples $\bar{a} \in D^c$ and $\bar{b} \in D^k$ we say that $\bar{b}$ is \emph{found} from $\bar{a}$ and $C$, if $c = |C|$, $\mathcal{D} \models \operatorname{Split}_{r}^C(\bar{b})$ and for every $i \in [c]$, the element $b_1^i$ (from $\bar{b}$) according to $C$, is equal to $a_i$. Intuitively, $\bar{a}$ consists of the most significant elements from $\bar{b}$ according to $C$.

\begin{rem}\label{remark: unique tuple and split} Let $\mathcal{D}$ be a $\sigma$-db and let $r,k \in \mathbb{N}$. For any $\bar{b} \in D^k$ there exists exactly one  $r$-split $C$ (of a tuple with $k$ elements) and tuple $\bar{a}$ from $\mathcal{D}$ such that $\bar{b}$ is found from $\bar{a}$ and $C$. \end{rem}

\begin{lem}\label{lemma: compute found tuple} Let $r,k,c \in \mathbb{N}$ where $c \leq k$. There exists an algorithm which, given a $\sigma$-db $\mathcal{D}$, a tuple $\bar{a} \in D^c$ and an $r$-split $C$ of a tuple with $k$ elements as input, returns a tuple  $\bar{b} \in D^k$ that is found from $\bar{a}$ and $C$ if one exists and returns false otherwise. Furthermore if such a $\bar{b}$ exists then it is unique.

The running time of the algorithm depends only on $r$, $|C|$, $k$, $\sigma$ and $d$.

%let $C$ be an $r$-split of a tuple with $k$ elements and let $\mathcal{D}$ be a $\sigma$-db. For any $c \leq k$ and $\bar{a} \in D^c$, there exists at most one $\bar{b} \in D^k$ that is found from $\bar{a}$ and $C$ and there exists an algorithm \polly{add input} that runs in constant time which returns $\bar{b}$ if it exists and returns false otherwise.
\end{lem}
\begin{proof}
Let  $\mathcal{D}$ be a $\sigma$-db, let $\bar{a} \in D^c$ and let $C$ be an $r$-split of a tuple with $k$ elements.
The following algorithm returns a tuple $\bar{b} \in D^k$ that is found from $\bar{a}$ and $C$ if one exists and returns false otherwise.
\begin{enumerate}
\item  If $|C| \neq c $ or $\mathcal{D} \not\models \bigwedge_{(C_i,F_i,\tau_i) \in C} \operatorname{sph}_{\tau_i}(a_i)$ then return false.
\item For each $i \in [c]$, let $\bar{b}^i$ be the tuple whose first element is $a_i$ such that $\mathcal{D} \models (\bar{b}^i = F_i(a_i))$. Then let $\bar{b}$ be the tuple found by combining all the $\bar{b}^i$ according to $C$.
\item If $\mathcal{D} \models \bigwedge_{1\leq i \neq j \leq c}(N_r(\bar{b}^i) \cap N_r(\bar{b}^j) = \emptyset)$, return $\bar{b}$. Otherwise, return false.
\end{enumerate}

The $3rk$-neighbourhood of an element can be computed in time only dependent on $r$, $k$, $\sigma$ and $d$. Hence Steps 1 and 2 runs in time dependent only on $r$, $k$, $\sigma$, $d$ and $|C|$ since each $\bar{b}^i$ can be found by exploring the $3rk$-neighbourhood of $a_i$. In Step 3, for every $i \in [c]$, $N_r(\bar{b}^i)$ can be computed in time only dependent on $r$, $|\bar{b}^i| \leq k$, $\sigma$ and $d$ and hence the running time of Step 3 depends only on $r$, $k$, $\sigma$, $d$ and $|C|$ also. Therefore the overall running time of the algorithm depends only on $r$, $k$, $\sigma$, $d$ and $|C|$ as required.

Assume a tuple $\bar{b}$ is returned by the above algorithm from $C$ and $\bar{a}$. Then clearly $c=|C|$, $\mathcal{D} \models \operatorname{Split}_{r}^C(\bar{b})$ and each $b_1^i$ according to $C$ is equal to $a_i$. Therefore $\bar{b}$ is found from $\bar{a}$ and $C$. 

Now assume that there does exist a tuple $\bar{b} \in D^k$ that is found from $\bar{a}$ and $C$. Then $\bar{b}$ is unique as there is only one way to choose each tuple $\bar{b}^i$ such that $\mathcal{D} \models (\bar{b}^i = F_i(a_i))$. Furthermore, it is easy to see that $\bar{b}$ will be outputted by the above algorithm.
%Takes constant time, there is at least one $C\in S$ where a b can be found for any a in V1. b is unique. every tuple can be found by only one pair.
\end{proof}

%\begin{lem}\label{lemma: unique found tuple} Let $r,k \in \mathbb{N}$ and let $\mathcal{D}$ be a $\sigma$-db. Then for any $\bar{b} \in D^k$ there exists exactly one  $r$-split $C$ (of a tuple with $k$ elements), $c \leq k$ and $\bar{a} \in D^c$ such that $\bar{b}$ is found from $\bar{a}$ and $C$.
%\end{lem}
%\begin{proof}
%
%\end{proof}

%We denote by $F(\bar{x}_a,\bar{x}_b)$ the fact that for all $1 \leq i \leq k-c$, $x_{b_i}$ is the $\alpha_i$-th element in the $rk$-neighbourhood of $\bar{x}_a$. We define the FO formula $\operatorname{Split}_{c,r}^L(\bar{x}):= \operatorname{sph}_{\tau}(\bar{x}_a) \land F(\bar{x}_a,\bar{x}_b)$.

\begin{lem}\label{lemma: equivalent partitions}
Let $T \subseteq T_r^{\sigma,d}(k)$. 
%Let $c$ be the maximum number of connected components in a type $\tau \in T$. 
We can compute a set of $r$-splits $S$ for $\bar{x}= (x_1,\dots,x_k)$ such that the following holds: For any $\sigma$-db $\mathcal{D}$ and tuple $\bar{a}\in D^k$, 
%if $\mathcal{D} \models \bigvee_{\tau \in T} \operatorname{sph}_{\tau}(\bar{a})$ then there exists exactly one $L \in \{L_1,\dots,L_m\}$ such that $\mathcal{D} \models \operatorname{Split}_{c,r}^{L}(\bar{a})$, and if $\mathcal{D} \models \bigvee_{1\leq i \leq m} \operatorname{Split}_{c,r}^{L_i}(\bar{a})$ then $\mathcal{D} \models \bigvee_{\tau \in T} \operatorname{sph}_{\tau}(\bar{a})$. 
$\mathcal{D} \models \bigvee_{\tau \in T} \operatorname{sph}_{\tau}(\bar{a})$ if and only if $\mathcal{D} \models \bigvee_{C \in S} \operatorname{Split}_{r}^{C}(\bar{a})$. 
\end{lem}

\begin{proof}
The algorithm proceeds as follows.
Let $S$ be an empty set. For each possible $r$-split $C=\{(C_1, F_1, \tau_1), \dots, (C_{\ell}, F_{\ell}, \tau_{\ell})\}$ of the tuple $\bar{x}$ do the following. Let $\mathcal{D}_0$ be the disjoint union of the fixed representatives of each $\tau_i$. Let $\bar{b} \in D_0^k$ be a tuple such that $\mathcal{D}_0 \models  \operatorname{Split}_{r}^{C}(\bar{b})$ (note that such a tuple exists by the definition of an $r$-split). Then if $\bar{b}$'s $r$-type in $\mathcal{D}_0$ is in $T$, add $C$ to $S$.

Towards correctness let $\mathcal{D}$ be a $\sigma$-db and let $\bar{a}\in D^k$. Let $C=\{(C_1, F_1, \tau_1), \dots, (C_{\ell}, F_{\ell}, \tau_{\ell})\}$ be the $r$-split such that $\mathcal{D} \models  \operatorname{Split}_{r}^{C}(\bar{a})$ (note that $C$ is unique by Remark \ref{remark: unique split}).  Let $\mathcal{D'}$ be the disjoint union of the fixed representatives of each $3rk$-type that appears in $C$ and let $\bar{b} \in D'^k$ be a tuple such that $\mathcal{D'} \models  \operatorname{Split}_{r}^{C}(\bar{b})$. It remains to show that $\mathcal{N}_r^{\mathcal{D}}(\bar{a}) \cong \mathcal{N}_r^{\mathcal{D'}}(\bar{b}) $. This completes the proof because by the construction of $S$, it implies that $C \in S$ if and only if the $r$-type of $\bar{a}$ in $\mathcal{D}$ is in $T$ (i.e. $\mathcal{D} \models \bigvee_{\tau \in T} \operatorname{sph}_{\tau}(\bar{a})$ if and only if $\mathcal{D} \models \bigvee_{C \in S} \operatorname{Split}_{r}^{C}(\bar{a})$). Recall that we use $a_j^i$ and $b_j^i$ to represent the elements from $\bar{a}$ and $\bar{b}$ respectively that are the elements from $C_i$ that appear $j$-th in the tuples $\bar{a}$ and $\bar{b}$ respectively. As $\mathcal{D} \models  \operatorname{Split}_{r}^{C}(\bar{a})$ and $\mathcal{D'} \models  \operatorname{Split}_{r}^{C}(\bar{b})$, by the definition of the formula $\operatorname{Split}_{r}^{C}(\bar{x})$, it follows that $\mathcal{N}_{3rk}^{\mathcal{D}}(a_1^i) \cong \mathcal{N}_{3rk}^{\mathcal{D'}}(b_1^i) $ for every $i \in [\ell]$. For every $i \in [\ell]$ and $j \in [|C_i|]$, $a_j^i$ is at distance at most $(2r +1)(|C_i|-1) \leq (2r +1)(k-1)\leq 3rk-r$ from $a_1^i$ in $\mathcal{D}$, and $b_j^i$ is at distance at most $(2r +1)(|C_i|-1) \leq (2r +1)(k-1)\leq 3rk-r$ from $b_1^i$ in $\mathcal{D'}$ (since each $F_i$ is $r$-good for $\tau_i$). Therefore for every $i \in [\ell]$, the $r$-neighbourhoods of $\bar{a}_i$ and $\bar{b}_i$ are contained in the $3rk$-neighbourhoods of $a_1^i$ and $b_1^i$ respectively and hence $\mathcal{N}_r^{\mathcal{D}}(\bar{a}_i) \cong \mathcal{N}_r^{\mathcal{D'}}(\bar{b}_i) $. Then since $N_r^{\mathcal{D}}(\bar{a}_i) \cap N_r^{\mathcal{D}}(\bar{a}_j)=\emptyset$ and $N_r^{\mathcal{D'}}(\bar{b}_i) \cap N_r^{\mathcal{D'}}(\bar{b}_j)=\emptyset$ (as $\mathcal{D} \models  \operatorname{Split}_{r}^{C}(\bar{a})$ and $\mathcal{D'} \models  \operatorname{Split}_{r}^{C}(\bar{b})$), it follows that $\mathcal{N}_r^{\mathcal{D}}(\bar{a}) \cong \mathcal{N}_r^{\mathcal{D'}}(\bar{b}) $.
\end{proof}

Let us first prove Theorem \ref{theorem: strengthened non local main}.

\begin{proof}[Proof of Theorem \ref{theorem: strengthened non local main}]

Let $\mathcal{D} \in \mathbf{C}_d^t$ with $|D|=n$, let $\epsilon \in (0,1]$ and let $\gamma \in (0,1)$. We shall construct an $\epsilon$-approximate enumeration algorithm for $\operatorname{Enum}_{\mathbf{C}_d^t}(\phi)$ that has answer threshold function $f(n)=\gamma n^c$, polylogarithmic preprocessing time and constant delay.

In the preprocessing phase, the algorithm starts by running the algorithm from Lemma \ref{lemma: computing rel r-types} on $\mathcal{D}$ to compute a set $T$ of $r$-types with $k$ centres. The algorithm then computes the set of $r$-splits $S$ from $T$ as in Lemma \ref{lemma: equivalent partitions}. An empty queue $\mathbf{Q}$ is then initialised which will store tuples to be outputted in the enumeration phase.

Let $V = \bigcup_{1 \leq i \leq c}D^i$. 
Let $V_1$ be the set that contains all $\bar{a} \in V$ such that there exists a $C \in S$ and $\bar{b} \in D^k$ where $\bar{b}$ is found from $\bar{a}$ and $C$.
%where $|C|=|\bar{a}|$ and there exists a $\bar{b} \in D^k$ such that $\mathcal{D} \models \operatorname{Split}_{r}^C(\bar{b})$ and for every $i \in [|C|]$, the element $b_1^i$ (from $\bar{b}$) according to $C$, is equal to $a_i$.
%$\mathcal{D} \models  \bigwedge_{(C_i,F_i,\tau_i) \in C} \operatorname{sph}_{\tau_i}(a_i) \wedge \bigwedge_{1\leq i \neq j \leq |C|}(N_r(\bar{a}^i) \cap N_r(\bar{a}^j) = \emptyset)$ where for every $i \in [|C|]$, $\bar{a}^i$ is the tuple from $\mathcal{D}$ such that $\mathcal{D} \models \bar{a}^i=F_i(a_i)$. 
Finally let $V_2 = V \setminus V_1$. Note that by Lemma \ref{lemma: compute found tuple}, given a tuple $\bar{a} \in V$ it can be decided in constant time whether $\bar{a} \in V_1$.

The algorithm from Lemma \ref{lemma: partitioned set enumeration} is then run with $\mu = \gamma/(c\cdot s(r,k))$, $\delta = 4/5$ and $V$, $V_1$ and $V_2$ as defined above. Once the enumeration phase of the algorithm from Lemma \ref{lemma: partitioned set enumeration} starts we do the following.
%The enumeration phase is then as follows:
\begin{enumerate}
\item Each time a tuple $\bar{a}$ is enumerated from the algorithm from Lemma \ref{lemma: partitioned set enumeration}, for each $C \in S$: run the algorithm from Lemma \ref{lemma: compute found tuple} with $\bar{a}$ and $C$ and if a tuple is returned add it to $\mathbf{Q}$.
\item If $\mathbf{Q} \neq \emptyset$, output the next tuple from $\mathbf{Q}$; stop otherwise.
\item Repeat Steps 1-2 until there is no tuple to output in step 2.
\end{enumerate}
%Each time a tuple $\bar{a}$ is outputted from the algorithm from Lemma \ref{lemma: partitioned set enumeration}, for each $C \in S$ such that $\mathcal{D} \models \bigwedge_{(C_i,F_i,\tau_i) \in C} \operatorname{sph}_{\tau_i}(a_i)$, 
%by exploring the $rk$-neighbourhood of $\bar{a}$,

%the tuple $\bar{a'}$ is computed such that $\mathcal{D} \models  \operatorname{Split}_{r}^{C}(\bar{a'})$ and 

From Lemma \ref{lemma: computing rel r-types} the set $T$ can be computed in polylogarithmic time. The set $S$ can be constructed in constant time as $|T|$ is a constant and the number of possible $r$-splits for a $k$-tuple is also a constant. Then as the preprocessing phase from the algorithm from Lemma \ref{lemma: partitioned set enumeration} runs in constant time the overall running time of the preprocessing phase is polylogarithmic. 

In the enumeration phase, by Lemma \ref{lemma: partitioned set enumeration} there is constant delay between the outputs of the tuples $\bar{a}$ used in Step 1. For every such tuple, by the definition of the set $V_1$, there exists at least one $r$-split in $S$ that leads to a tuple being added to $\mathbf{Q}$. Then as $|S|$ is a constant and the algorithm from Lemma \ref{lemma: compute found tuple} runs in constant time, the enumeration phase has constant delay as required. This concludes the analysis of the running time. Let us now prove correctness.

By Lemma \ref{lemma: partitioned set enumeration} in Step 1 of the enumeration phase no duplicate tuples $\bar{a}$ will be considered. Since for every tuple $\bar{b} \in D^k$ there exists exactly one $r$-split $C$ and tuple $\bar{a}$ from $\mathcal{D}$ such that $\bar{b}$ is found from $C$ and $\bar{a}$ (Remark \ref{remark: unique tuple and split}), no duplicates will be enumerated.

Now let us assume that the set of $r$-types $T$ were computed correctly (i.e. for any $\bar{a} \in D^k$,
 if $\bar{a} \in \phi(\mathcal{D})$, then the $r$-type of $\bar{a}$ in $\mathcal{D}$ is in $T$, and if $\bar{a} \in D^k \setminus \phi(\mathcal{D},\mathbf{C}_d^t, \epsilon )$, then the $r$-type of $\bar{a}$ in $\mathcal{D}$ is not in $T$) which happens with probability at least $5/6$ by Lemma \ref{lemma: computing rel r-types}.
%Therefore by Lemma \ref{lemma: equivalent partitions} for every $\bar{b} \in D^k$, let $C$ be the $r$-split such that  $\mathcal{D} \models \operatorname{Split}_{r}^C(\bar{b})$, then if $\bar{b} \in \phi(\mathcal{D})$, $C$ is in $S$, and if $\bar{a} \in D^k \setminus \phi(\mathcal{D},\mathbf{C}_d^t, \epsilon )$, then $C$ is not in $S$. 
Let $\bar{b} \in D^k$ have $r$-type $\tau$ in $\mathcal{D}$ and let $C \in S_r^{\sigma,d}(k)$ be such that $\mathcal{D} \models \operatorname{Split}_{r}^{C}(\bar{b})$. 

If $\bar{b} \in D^k \setminus \phi(\mathcal{D},\mathbf{C}_d^t, \epsilon )$, $\tau \not\in T$ and hence by Lemma \ref{lemma: equivalent partitions}, $C \not\in S$ and so $\bar{b}$ will not be enumerated. Therefore with probability at least $5/6$ only tuples from $\phi(\mathcal{D},\mathbf{C}_d^t, \epsilon )$ will be enumerated.

If $\bar{b} \in \phi(\mathcal{D})$, then $\tau \in T$ and hence by Lemma \ref{lemma: equivalent partitions}, $C \in S$. Let $\bar{a}$ be the tuple such that $\bar{b}$ is found from $\bar{a}$ and $C$. Note that as the maximum number of connected components in the $r$-neighbourhood of $\bar{b}$ in $\mathcal{D}$ is $c$, $|\bar{a}| \leq c$ and hence $\bar{a} \in V$. Then by definition $\bar{a} \in V_1$. Hence if every tuple from $V_1$ is considered in Step 1 of the enumeration phase, every tuple in $\phi(\mathcal{D})$ will be enumerated. By Lemma \ref{lemma: partitioned set enumeration} with probability at least $\delta$ if $|V_1| \geq \mu |V|$, every tuple from $V_1$ will be considered in Step 1. We know that $|V_1| \geq |\phi(\mathcal{D})|/s(r,k)$ as every $\bar{a} \in V_1$ leads us to at most $|S| \leq s(r,k)$ many tuples from $\phi(\mathcal{D})$ (since by Lemma \ref{lemma: compute found tuple} for any $r$-split $C\in S$ there is at most one tuple that is found from $\bar{a}$ and $C$). If $|\phi(\mathcal{D})| \geq \gamma n^c$ then $|V_1| \geq \gamma n^c/s(r,k) \geq \gamma |V|/(c\cdot s(r,k)) = \mu |V|$ as $|V| = \sum_{i=1}^c n^i \leq c n^c$ and by the choice of $\mu$. Hence if $|\phi(\mathcal{D})| \geq \gamma n^c$ with probability at least $\delta \cdot 5/6 = 2/3$ every tuple from $\phi(\mathcal{D})$ will be enumerated. This completes the proof.
\end{proof}

We now prove Theorem \ref{theorem: strengthened local main} which is similar to the proof of Theorem \ref{theorem: strengthened non local main}.

\begin{proof}[Proof of Theorem \ref{theorem: strengthened local main}] 
First let us note that if $\phi$ is local then by Lemma \ref{lemma: local query rel r-types} we can compute a set $T$ of $r$-types (where $r$ is the locality radius of $\phi$) in constant time such that for any $\sigma$-db $\mathcal{D}$ and tuple $\bar{a}$ from $\mathcal{D}$, the $r$-type of $\bar{a}$ in $\mathcal{D}$ is in $T$ if and only if $\bar{a} \in \phi(\mathcal{D})$. 

Then to construct an algorithm as in the theorem statement we can just use the algorithm from the proof of Theorem \ref{theorem: strengthened non local main} but change it in two ways. Firstly we allow the input class to be any class of bounded degree $\sigma$-dbs and secondly, we construct $T$ as discussed above. The only part of the algorithm from the proof of Theorem \ref{theorem: strengthened non local main} that runs in non-constant time is the construction of $T$ and hence our algorithm has the required running times. 

To prove correctness first note that in the proof of Theorem \ref{theorem: strengthened non local main} the only reason the input class was $\mathbf{C}_d^t$ was to allow the set $T$ to be computed efficiently and with high probability correctly. Now $T$ is computed exactly and since the algorithm will only enumerate tuples that have $r$-type in $T$, only tuples that are answers to the query for the input database will be enumerated as required. The proof of (2) from the theorem statement is then very similar to the last paragraph in the proof of Theorem \ref{theorem: strengthened non local main} (the only difference is that now for local queries this happens with higher probability as $T$ is computed exactly every time).
% Then the proof of Theorem \ref{theorem: strengthened local main} is the same as the proof of Theorem \ref{theorem: strengthened non local main} apart from in the construction of $T$ (which can be constructed in constant time and exactly now) and the input class can be any class of bounded degree $\sigma$-dbs (in the proof of Theorem \ref{theorem: strengthened non local main} the only reason the input class was $\mathbf{C}_d^t$ was to allow the set $T$ to be computed efficiently and with high probability correctly).
\end{proof}

% then for any $\sigma$-db $\mathcal{D}$, $\epsilon \in (0,1]$ and input class $\mathbf{C}$, $\phi(\mathcal{D}) = \phi(\mathcal{D}, \mathbf{C}, \epsilon)$. 

\section{Further Results}\label{sec: further results}

In this section, we start by generalising our result on approximate enumeration of general FO queries (Theorem~\ref{theorem: strengthened non local main}). We identify a condition that we call \emph{Hanf-sentence testability}, which is a weakening of the bounded tree-width condition, under which we still get approximate enumeration algorithms with the same probabilistic guarantees as before. Finally, we discuss approximation versions of query membership testing and counting.

\subsection{Generalising Theorem \ref{theorem: strengthened non local main}}

We first introduce Hanf-sentence testability, which is based on the Hanf normal form of a formula. It allows us to compute the set of $r$-types as in Lemma \ref{lemma: computing rel r-types} efficiently. Theorem~\ref{theorem: strengthened main hanf} below is the generalisation of Theorem \ref{theorem: strengthened non local main}, and Example~\ref{running example 2} illustrates the use of this generalisation.

\begin{defi}[Hanf-sentence testable]
Let $\phi(\bar{x}) \in \operatorname{FO}[\sigma]$ and $\chi(\bar{x})$ be the formula in the form (\ref{normal-form}) of Lemma \ref{normal-form lemma} that is $d$-equivalent to $\phi$. Let $m$ be the number of conjunctive clauses in $\chi$. We say that $\phi$ is \emph{Hanf-sentence testable on $\mathbf{C}$ in time $H(n)$} if for every $i \in [m]$, the formula $\exists \bar{x} \operatorname{sph}_{\tau_i}(\bar{x}) \land \psi_i^s$ is uniformly testable on $\mathbf{C}$ in time at most $H(n)$.
\end{defi}

We shall illustrate Hanf sentence testability in the following example. 

\begin{exa}\label{running example Hanf testable}
	Let $\phi$ be as in Example \ref{running example} and let $\mathcal{G}\in \mathbf{G}_d$. If there exists $(u,v) \in V(\mathcal{G})^2$ with $2$-type $\tau_1$ then there exists a vertex with $2$-type $\tau_4$ (where $\tau_4$ is as in Example \ref{tester example}) and vice versa. Hence, $\phi$ can be easily transformed into the form (\ref{normal-form}) of Lemma \ref{normal-form lemma} by replacing the subformula $\lnot(\exists z \exists w\operatorname{sph}_{\tau_1}(z,w))$ with $\lnot \exists ^{\geq 1}z\operatorname{sph}_{\tau_4}(z)$. The resulting formula then has two conjunctive clauses, ${sph}_{\tau_2}(x,y) \land \lnot \exists ^{\geq 1}z\operatorname{sph}_{\tau_4}(z)$ and ${sph}_{\tau_1}(x,y)$.
	 We saw in Example~\ref{tester example} that $\exists x \exists y \operatorname{sph}_{\tau_2}(x,y) \land \lnot \exists ^{\geq 1}z\operatorname{sph}_{\tau_4}(z)$ is uniformly testable on $\mathbf{G}_d$ in constant time. The formula $\exists x \exists y \operatorname{sph}_{\tau_1}(x,y)$ is trivially testable in constant time on $\mathbf{G}_d$ since we can insert a copy of $\tau_1$ into a graph $\mathcal{G} \in \mathbf{G}_d$ with at most $8d +7$ modifications and therefore if $8d+7 \leq \epsilon d |V(\mathcal{G})|$ we can always accept and otherwise (i.e. if $|V(\mathcal{G})| < (8d+7)/\epsilon d$) we can just do a full check of the graph for a copy of $\tau_1$ in constant time. Hence, $\phi$ is Hanf-sentence testable on $\mathbf{G}_d$ in constant time.
\end{exa}

%\begin{thm}\label{thm:non-local enumeration hanf-sentence}
%Let $\phi(\bar{x}) \in \operatorname{FO}[\sigma]$ where $|\bar{x}|=k$. If $\phi$ is Hanf-sentence testable on $\mathbf{C}$ in time $H(n)$, then $\operatorname{Enum}_{\mathbf{C}}(\phi)$ can be solved approximately with preprocessing time $\mathcal{O}(H(n))$ and constant delay for answer threshold function $f(n)=\gamma n^k$ for any parameter $\gamma \in (0,1)$. 
%\end{thm}

Note that any FO query is Hanf sentence testable on $\mathbf{C}_d^t$ in polylogarithmic time.
We shall now prove a result that is similar to Lemma \ref{lemma: computing rel r-types} but works for any class $\mathbf{C}$ and FO query $\phi$ where $\phi$ is Hanf-sentence testable on $\mathbf{C}$. This will then be used to show we can replace bounded tree-width with Hanf sentence testability and still obtain enumeration algorithms with the same probabilistic guarantees.

\begin{lem}\label{lemma: computing rel r-types strengthened}
Let $\phi(\bar{x}) \in \operatorname{FO}[\sigma]$ with $|\bar{x}|=k$ and Hanf locality radius $r$ and let $\epsilon \in (0,1]$. If $\phi$ is Hanf-sentence testable on $\mathbf{C}$ in time $H(n)$ then there exists an algorithm $\mathbb{B}_{\epsilon}$ that runs in time $\mathcal{O}(H(n))$, which, given oracle access to a $\sigma$-db $\mathcal{D} \in \mathbf{C}$ as input along with $|D|=n$, computes a set $T$ of $r$-types with $k$ centres such that with probability at least $5/6$, for any $\bar{a} \in D^k$,
\begin{enumerate}
\item if $\bar{a} \in \phi(\mathcal{D})$, then the $r$-type of $\bar{a}$ in $\mathcal{D}$ is in $T$, and
\item if $\bar{a} \in D^k \setminus \phi(\mathcal{D},\mathbf{C}, \epsilon )$, then the $r$-type of $\bar{a}$ in $\mathcal{D}$ is not in $T$.
\end{enumerate} 
\end{lem}
\begin{proof} The algorithm $\mathbb{B}_{\epsilon}$ is nearly identical to the algorithm $\mathbb{A}_{\epsilon}$ from Lemma \ref{lemma: computing rel r-types}. The only difference being is we replace the input class $\mathbf{C}_d^t$ with $\mathbf{C}$. The $\epsilon/2$-testers $\pi_i$ used now have input class $\mathbf{C}$ (rather than $\mathbf{C}_d^t$) and as $\phi$ is Hanf-sentence testable on $\mathbf{C}$ in time $H(n)$ each $\pi_i$ runs in time $\mathcal{O}(H(n))$ (rather than polylogarithmic). As all other parts of the algorithm $\mathbb{A}_{\epsilon}$ run in constant time, it follows that $\mathbb{B}_{\epsilon}$ runs in time $\mathcal{O}(H(n))$ as required. The proof of the correctness of $\mathbb{B}_{\epsilon}$ is then identical to the proof of the correctness of $\mathbb{A}_{\epsilon}$ (but with the input class $\mathbf{C}_d^t$ replaced with $\mathbf{C}$).
\end{proof}

We shall now show that if a FO query $\phi$ is Hanf-sentence testable on a class $\mathbf{C}$ in time $H(n)$ then $\operatorname{Enum}_{\mathbf{C}}(\phi)$ can be solved approximately with preprocessing time $\mathcal{O}(H(n))$ and constant delay. Note we are still able to reduce the answer threshold function.

\begin{thm}\label{theorem: strengthened main hanf} 
Let $\phi(\bar{x}) \in \operatorname{FO}[\sigma]$ and let $c : =\operatorname{conn}(\phi, d)$.
 If $\phi$ is Hanf-sentence testable on $\mathbf{C}$ in time $H(n)$, then $\operatorname{Enum}_{\mathbf{C}}(\phi)$ can be solved approximately with preprocessing time $\mathcal{O}(H(n))$ and constant delay for answer threshold function $f(n)=\gamma n^{c}$ for any $\gamma \in (0,1)$.
\end{thm}
\begin{proof} 
Let $\epsilon \in (0,1]$, let $\gamma \in (0,1)$ and let us assume that $\phi$ is Hanf-sentence testable on $\mathbf{C}$ in time $H(n)$. If we take the $\epsilon$-approximate enumeration algorithm for $\operatorname{Enum}_{\mathbf{C}_d^t}(\phi)$ with answer threshold function $f(n)=\gamma n^{c}$ given in the proof of Theorem \ref{theorem: strengthened non local main}, which we shall denote by $\mathbb{E}_{\phi, \mathbf{C}_d^t, \epsilon} $, and make the following changes: replace the input class $\mathbf{C}_d^t$ with $\mathbf{C}$, and use Lemma \ref{lemma: computing rel r-types strengthened} instead of Lemma \ref{lemma: computing rel r-types} to compute the set of $r$-types $T$. Then we argue that the resulting algorithm $\mathbb{E}_{\phi, \mathbf{C}, \epsilon} $ is an $\epsilon$-approximate enumeration algorithm for $\operatorname{Enum}_{\mathbf{C}}(\phi)$ with preprocessing time $\mathcal{O}(H(n))$ and constant delay for answer threshold function $f(n)=\gamma n^{c}$.

In the preprocessing phase of $\mathbb{E}_{\phi, \mathbf{C}_d^t, \epsilon} $ the only part that runs in non-constant time is the construction of the set $T$ (which takes polylogarithmic time). In $\mathbb{E}_{\phi, \mathbf{C}, \epsilon} $ it takes $\mathcal{O}(H(n))$ time to compute $T$ and hence $\mathbb{E}_{\phi, \mathbf{C}, \epsilon} $ has preprocessing time $\mathcal{O}(H(n))$.
Since $\mathbb{E}_{\phi, \mathbf{C}_d^t, \epsilon} $ has constant delay, $\mathbb{E}_{\phi, \mathbf{C}, \epsilon} $ also has constant delay.

Since the only differences in Lemma \ref{lemma: computing rel r-types} and Lemma \ref{lemma: computing rel r-types strengthened} is the running times and the input class, the proof of the correctness of $\mathbb{E}_{\phi, \mathbf{C}, \epsilon} $ is the same as the proof of the correctness of $\mathbb{E}_{\phi, \mathbf{C}_d^t, \epsilon} $ but with the input class $\mathbf{C}_d^t$ replaced with $\mathbf{C}$.
\end{proof}

We shall now return to our running example where we discuss an FO query and input class, which previous theorems did not give us an approximate enumeration algorithm for, but by Theorem \ref{theorem: strengthened main hanf} can now be approximately enumerated.

\begin{exa}\label{running example 2}
Let $\phi$ be the formula as in Example \ref{running example}.
We saw in Example \ref{running example Hanf testable} that $\phi$ is Hanf-sentence testable on $\mathbf{G}_d$ in constant time and that the formula in the form (\ref{normal-form}) of Lemma \ref{normal-form lemma} that is $d$-equivalent to $\phi$ is $\chi(x,y) =\operatorname{sph}_{\tau_1}(x,y) \lor (\operatorname{sph}_{\tau_2}(x,y) \land \lnot \exists ^{\geq 1}z\operatorname{sph}_{\tau_4}(z))$. The maximum number of connected components of the neighbourhood types that appear in the sphere-formulas of $\chi$ is one. Hence, by Theorem \ref{theorem: strengthened main hanf}, $\operatorname{Enum}_{\mathbf{G}_d}(\phi)$ can be solved approximately with constant preprocessing time and constant delay for answer threshold function $f(n)=\gamma n$ for any parameter $\gamma \in (0,1)$. 
\end{exa}

%Add membership testing and counting here? -yes but not too much

\subsection{Approximate query membership testing}
The \emph{query membership testing problem} for $\phi(\bar{x}) \in \operatorname{FO}[\sigma]$ over $\mathbf{C}$ is the computational problem where, for a database $\mathcal{D} \in \mathbf{C}$, we ask whether a given tuple $\bar{a} \in D^k$ satisfies $\bar{a}\in \phi(\mathcal{D})$. We call $\bar{a}$ the \emph{dynamical input} and the answer (`true' or `false') the \emph{dynamical answer}. 
Similar to query enumeration, the goal is to obtain an algorithm, that, after a preprocessing phase, can answer membership queries for dynamical inputs very efficiently. The preprocessing phase should also
be very efficient.
Kazana \cite{kazana2013query} shows that the query membership testing problem for any $\phi(\bar{x}) \in \operatorname{FO}[\sigma]$ over $\mathbf{C}$ can be solved by an algorithm with a linear time preprocessing phase, and an answering phase that, for a given dynamical input, computes the dynamical answer in constant time.

Given a local FO query, by Lemma \ref{lemma: local membership testing}, for any $\sigma$-db $\mathcal{D}$ and tuple $\bar{a}$ from $\mathcal{D}$ we can test in constant time whether $\bar{a} \in \phi (\mathcal{D})$. Hence in this section we shall focus on general queries.

We introduce an approximate version of the query membership testing problem. We say that the query membership testing problem for $\phi(\bar{x}) \in \operatorname{FO}[\sigma]$ over $\mathbf{C}$ can be \emph{solved approximately} with an $\mathcal{O}(H(n))$-time preprocessing phase and constant-time answering phase if for any $\epsilon \in (0,1]$, there exists an algorithm, which is given oracle access to a database $\mathcal{D} \in \mathbf{C}$ and $|D| =n$ as an input, and proceeds in two phases.
\begin{enumerate}
\item A preprocessing phase that runs in time $\mathcal{O}(H(n))$.
\item An answer phase where, given dynamical input $\bar{a} \in D^k$, 
	%in constant time, with probability at least ${2}/{3}$
	the following is computed in constant time.
\begin{itemize}
\item If $\bar{a} \in \phi(\mathcal{D})$, the algorithm returns `true', with probability at least $2/3$, and
\item if $\bar a\notin \phi(\mathcal D,\mathbf C, \epsilon)$, the algorithm returns `false', with probability at least $2/3$.
\end{itemize}
\end{enumerate}
The following follows from the proof of Lemma~\ref{lemma: computing rel r-types}.

\begin{thm}\label{query membership testing}
The query membership testing problem for $\phi(\bar{x}) \in \operatorname{FO}[\sigma]$ (where $|\bar{x}|=k$) over $\mathbf{C}^t_d$ can be solved approximately with a polylogarithmic preprocessing phase and constant-time answering phase.
%The query testing problem for $\phi(\bar{x}) \in \operatorname{FO}[\sigma]$ over $\mathbf{C}$ can be solved approximately with an $\mathcal{O}(H(n))$-time preprocessing phase and constant-time answering phase if $\phi$ is Hanf-sentence testable on $\mathbf{C}$ in time $H(n)$.
\end{thm}
\begin{proof}
Let $r$ be the Hanf locality radius of $\phi$. In the preprocessing phase a set $T$ of $r$-types as in Lemma \ref{lemma: computing rel r-types} is computed. Then in the answer phase, given a tuple $\bar{a} \in D^k$, the $r$-type $\tau$ of $\bar{a}$ is computed. If $\tau \in T$ then the algorithm returns `true', otherwise it returns `false'. By Lemma \ref{lemma: computing rel r-types} the set $T$ can be computed in polylogarithmic time and it takes constant time to calculate $\tau$. By Lemma \ref{lemma: computing rel r-types} with probability at least $5/6 >2/3$, if $\bar{a} \in \phi(\mathcal{D})$, then the $r$-type of $\bar{a}$ in $\mathcal{D}$ is in $T$, and
 if $\bar{a} \in D^k \setminus \phi(\mathcal{D},\mathbf{C}_d^t, \epsilon )$, then the $r$-type of $\bar{a}$ in $\mathcal{D}$ is not in $T$. Therefore with probability at least $2/3$ if $\bar{a} \in \phi(\mathcal{D})$ the algorithm outputs `true' and if $\bar a\notin \phi(\mathcal D,\mathbf C, \epsilon)$ the algorithm outputs `false' as required. 
\end{proof}

Note that we can get a similar result for Hanf sentence testable FO queries over any class of bounded degree graphs.
%From Theorem \ref{query membership testing} and Theorem \ref{FO testability} the following follows. 
%
%\begin{cor}\label{cor:testing}
%The query membership testing problem for $\phi(\bar{x}) \in \operatorname{FO}[\sigma]$ over $\mathbf{C}^t_d$ can be solved approximately with a polylogarithmic preprocessing phase and constant-time answering phase.
%\end{cor}

\subsection{Approximate counting}
The \emph{counting problem} for $\phi(\bar{x}) \in \operatorname{FO}[\sigma]$ over $\mathbf{C}$ is the problem of, given a database $\mathcal{D} \in \mathbf{C}$, compute $|\phi(\mathcal{D})|$.  It was shown in \cite{bagan2008computing} that the counting problem for any $\phi(\bar{x}) \in \operatorname{FO}[\sigma]$ over $\mathbf{C}$ can be solved in linear time.

Lemma~5.1 in~\cite{newman2013every}, allows approximating the distribution of the $r$-types with one centre of an input graph by looking at a constant number of vertices. We can easily extend this to databases and neighbourhood types with multiple centres.

%On $\sigma$-dbs of degree at most $d$, the number of $r$-types with $k$ centres is finite for fixed $r$ and $k$. Let 
% $\operatorname{c}(r,k)$ 
% be the number of $r$-types with $k$ centres, and fix an enumeration $\tau_1, \ldots,\tau_{\operatorname{c}(r,k)}$ of the $r$-types.
We fix an enumeration $\tau_1, \ldots,\tau_{\operatorname{c}(r,k)}$ of the $r$-types in $T_{r}^{\sigma, d}(k)$.
 For a $\sigma$-db $\mathcal D$ with $|D|=n$, the \emph{$k$ centre $r$-neighbourhood distribution of $\mathcal D$} 
 is the vector 
 $\dv_{r,k}(\mathcal D)$ of length $\operatorname{c}(r,k)$ whose $i$-th component (denoted by $\dv_{r,k}(\mathcal D)[i]$) contains the number $t(i)/n^k$, where 
 $t(i)\in \mathbb N$ is
 the number of elements of $D$ 
 whose $r$-type is $\tau_i$. 

We let $\operatorname{EstimateFrequencies}_{r,s,k}$ be an algorithm with oracle access to an input database $\mathcal{D} \in \mathbf{C}$, that samples $s$ tuples from $D^k$ uniformly and independently and explores their $r$-neighbourhoods. $\operatorname{EstimateFrequencies}_{r,s,k}$ returns the distribution vector $\bar{v}$ of the $r$-types of this sample.
$\operatorname{EstimateFrequencies}_{r,s,k}$ has constant running time, independent of $|D|$, and comes with the following guarantees.

\begin{lem}\label{neighbourhood distribution}
	Let $\mathcal{D} \in \mathbf{C}$ be a database on $n$ elements, $\lambda \in (0,1)$ and $r,k \in \mathbb{N}$. If $s \geq \operatorname{c}(r,k)^2/\lambda^2 \cdot \operatorname{ln}(20\operatorname{c}(r,k))$, with probability at least $9/10$ the vector $\bar{v}$ returned by \\$\operatorname{EstimateFrequencies}_{r,s,k}$ on input $\mathcal{D}$ satisfies $\| \bar{v} - \dv_{r,k}(\mathcal{D}) \|_1 \leq \lambda$.
\end{lem} 
By combining Lemmas \ref{lemma: computing rel r-types} and \ref{lemma: equivalent partitions} and Lemma \ref{neighbourhood distribution} we get the following result.
%
%\begin{thm}\label{counting theorem}
%Let $\phi(\bar{x}) \in \operatorname{FO}[\sigma]$ with $k$ free variables and $\mathcal{D} \in \mathbf{C}$ be a $\sigma$-db on $n$ elements. Let $\epsilon \in (0,1]$ and $\lambda \in (0,1)$. If $\phi$ is Hanf-sentence testable on $\mathbf{C}$ in time $H(n)$ then there exists an algorithm, which is given oracle access to $\mathcal{D}$ and $|D|=n$ as an input, that returns an estimate of $|\phi(\mathcal{D})|$ such that with probability at least ${2}/{3}$ the estimate is within the range
%$[|\phi(\mathcal{D})| - \lambda n^k, |\phi(\mathcal{D}) \cup \phi(\mathcal{D}, \mathbf{C}, \epsilon)| + \lambda n^k].$
%Furthermore, the algorithm runs in time $\mathcal{O}(H(n))$.
%\end{thm}
%\begin{proof}
%	Let $r$ be the locality radius of $\phi$ and $s = \operatorname{c}(r,k)^2/\lambda^2 \cdot \operatorname{ln}(20\operatorname{c}(r,k))$. The algorithm first runs $\operatorname{EstimateFrequencies}_{r,s,k}$ on input $\mathcal{D}$. Let $\bar{v}$ be the vector returned. Next, the algorithm computes a set Rel of $r$-types with $k$ centres as in Lemma \ref{relevant set}. The algorithm then returns \[n^k \sum_{\tau_i \in \operatorname{Rel}} v[i].\] 
%By Lemma \ref{relevant set} and Lemma \ref{neighbourhood distribution} the overall running time is $\mathcal{O}(H(n))$ and by looking at the two extreme cases with probability at least $9/10 \cdot 5/6 > 2/3$
% \[|\phi(\mathcal{D})| - \lambda n^k \leq n^k \sum_{\tau_i \in \operatorname{Rel}} v[i] \leq |\phi(\mathcal{D}) \cup \phi(\mathcal{D}, \mathbf{C}, \epsilon)| + \lambda n^k\]
%	as required.
%\end{proof}

\begin{thm}\label{counting theorem}
Let $\phi(\bar{x}) \in \operatorname{FO}[\sigma]$, let $\epsilon \in (0,1]$, let $\lambda \in (0,1)$ and let $c : =\operatorname{conn}(\phi, d)$. 
There exists an algorithm, which, given oracle access to $\mathcal{D} \in \mathbf{C}_d^t$ and $|D|=n$ as an input, returns an estimate of $|\phi(\mathcal{D})|$ such that with probability at least ${2}/{3}$ the estimate is within the range
$[|\phi(\mathcal{D})| - \lambda cn^c, |\phi(\mathcal{D}) \cup \phi(\mathcal{D}, \mathbf{C}^t_d, \epsilon)| + \lambda cn^c].$
	Furthermore, the algorithm runs in polylogarithmic time in $n$.

%Let $\phi(\bar{x}) \in \operatorname{FO}[\sigma]$ and $\mathcal{D} \in \mathbf{C}$ be a $\sigma$-db on $n$ elements. Let $\epsilon \in (0,1]$, $\lambda \in (0,1)$ and  $c \in [1,k]$ be the maximum number of connected components in a neighbourhood type in $\operatorname{Rel}(\mathcal{D}, \phi)\cup \operatorname{Rel}^{\mathbf{C}}_{\epsilon\text{-close}}(\mathcal{D}, \phi)$ for any $\mathcal{D} \in \mathbf{C}$. If $\phi$ is Hanf-sentence testable on $\mathbf{C}$ in time $H(n)$ then there exists an algorithm, which is given oracle access to $\mathcal{D}$ and $|D|=n$ as an input, that returns an estimate of $|\phi(\mathcal{D})|$ such that with probability at least ${2}/{3}$ the estimate is within the range
%$[|\phi(\mathcal{D})| - \lambda n^{c}, |\phi(\mathcal{D}) \cup \phi(\mathcal{D}, \mathbf{C}, \epsilon)| + \lambda n^{c}].$
%Furthermore, the algorithm runs in time $\mathcal{O}(H(n))$.
\end{thm}

We shall only give the proof idea of Theorem \ref{counting theorem}. To estimate $|\phi(\mathcal{D})|$ we can do the following. We will start by computing a set of $r$-types $T$ as in Lemma \ref{lemma: computing rel r-types} and then use Lemma \ref{lemma: equivalent partitions} to compute the set of $r$-splits $S$ from $T$. Then using Lemma \ref{neighbourhood distribution}, for every $i \in [c]$ we can compute an estimate to the vector $\dv_{3rk,i}(\mathcal D)$. For every $i \in [c]$, we can also compute a vector $\bar{v}_i$ which has a component corresponding to each $3rk$-type $\tau$ with $i$ centres. The component in $\bar{v}_i$ that corresponds to the $3rk$-type $\tau$ is the number of $r$-splits $C \in S$ such that for any tuple $\bar{a}$ in $\mathcal{D}$ with $3rk$-type $\tau$ there will exist exactly one tuple that is found from $\bar{a}$ and $C$. We can then estimate $|\phi(\mathcal{D})|$ using the estimates to the vectors $\dv_{3rk,i}(\mathcal D)$ and the vectors $\bar{v}_i$. By Lemmas \ref{lemma: computing rel r-types} and \ref{neighbourhood distribution} this estimate can be computed in polylogarithmic  time. 

 For correctness, first note that by Remark \ref{remark: unique tuple and split} for any tuple $\bar{b}$ in $\mathcal{D}$ there exists exactly one tuple $\bar{a}$ from $\mathcal{D}$ and one $r$-split $C$ such that $\bar{b}$ is found from $\bar{a}$ and $C$. Therefore we will not double count any tuple. By Lemmas \ref{lemma: computing rel r-types} and \ref{lemma: equivalent partitions} and the construction of the vectors $\bar{v}_i$, with high probability we will get an estimation to the number of tuples in $\phi(\mathcal{D})$.  By looking at the two extreme cases, where $T$ contains only $r$-types of tuples in $\phi(\mathcal{D})$, and $T$ contains all $r$-types of tuples in $\phi(\mathcal{D}) \cup \phi(\mathcal{D}, \mathbf{C}, \epsilon)$, it is easy to see that the returned estimate will be within the desired range.

The obvious limitation of Theorem~\ref{counting theorem} is that 
$|\phi(\mathcal{D}) \cup \phi(\mathcal{D}, \mathbf{C}^t_d, \epsilon)|$ can be much larger than
$|\phi(\mathcal{D})|$,
as discussed in Example~\ref{ex:many-eps-close}. Nevertheless, in application where the focus is on structural
closeness and very efficient running time, this might be tolerable.

\subsection*{Acknowledgement}
%\textbf{Acknowledgement.}
We would like to thank Benny Kimelfeld for inspiring discussions during early stages of this work.

\bibliographystyle{plainurl}
\bibliography{approx_query_enumeration_rewrite}

\end{document}